\documentclass[onefignum,onetabnum]{siamart190516}

\usepackage{amsfonts}
\usepackage{graphicx}
\usepackage{epstopdf}
\usepackage{algorithmic}
\ifpdf
  \DeclareGraphicsExtensions{.eps,.pdf,.png,.jpg}
\else
  \DeclareGraphicsExtensions{.eps}
\fi
\usepackage{graphicx}
\usepackage{color}
\usepackage{caption}
\usepackage{subcaption}
\usepackage{bbm}
\usepackage{amssymb}
\usepackage{mathtools}

\usepackage{soul} 
\usepackage{stmaryrd} 
\usepackage{algorithm} 
\usepackage{algorithmic}
\usepackage[colorinlistoftodos,textsize=tiny]{todonotes}
\newcommand{\Comments}{1}
\newcommand{\mynote}[2]{\ifnum\Comments=1\textcolor{#1}{#2}\fi}
\newcommand{\mytodo}[2]{\ifnum\Comments=1%
  \todo[linecolor=#1!80!black,backgroundcolor=#1,bordercolor=#1!80!black]{#2}\fi}

\newsiamremark{remark}{Remark}
\newsiamremark{hypothesis}{Hypothesis}
\crefname{hypothesis}{Hypothesis}{Hypotheses}
\newsiamthm{claim}{Claim}

\newcommand{\HH}{\mathbb{H}}

\newcommand{\RR}{\mathbb{R}}
\newcommand{\Lev}{\mathbb{L}}
\newcommand{\Ham}{\HH}
\newcommand{\abs}[1]{\lvert#1\rvert}
\newcommand{\indicator}[1][\ ]{ [\![ {#1} ]\!] \phantom{]}}

\headers{Levenshtein Graphs}{Perrin Ruth, and Manuel E. Lladser}

\title{Levenshtein Graphs: Resolvability, Automorphisms \& Determining Sets\thanks{Submitted to the editors DATE.
\funding{This work has been partially funded by the NSF grant No. 1836914}}}

\author{Perrin E. Ruth
\and Manuel E. Lladser\thanks{Department of Applied Mathematics, University of Colorado, Boulder, CO 80309
  (\email{manuel.lladser@colorado.edu}).}
  }
\usepackage{amsopn}

\ifpdf
\hypersetup{
  pdftitle={Levenshtein Graphs: Resolvability, Automorphisms \& Determining Sets},
  pdfauthor={Perrin Ruth, and Manuel E. Lladser}
}
\fi

\begin{document}

\maketitle

\begin{abstract}
We introduce the notion of Levenshtein graphs, an analog to Hamming graphs but using the edit distance instead of the Hamming distance; in particular, Levenshtein graphs allow for underlying strings (nodes) of different lengths. We characterize various properties of these graphs, including a necessary and sufficient condition for their geodesic distance to be identical to the edit distance, their automorphism group and determining number, and an upper bound on their metric dimension. Regarding the latter, we construct a resolving set composed of two-run strings and an algorithm that computes the edit distance between a string of length $k$ and any single-run or two-run string in $O(k)$ operations.
\end{abstract}

\begin{keywords}
edit distance, graph embedding, Hamming graph, Levenshtein graph, multilateration, node2vec, resolving set
\end{keywords}

\begin{AMS}
05C12, 05C85, 68R10, 68W32
\end{AMS}

\section{Introduction}

For a general unweighted graph $G = (V,E)$, a set $R\subset V$ is called resolving when for all $u,v\in V$, if $d(u,r)=d(v,r)$ for each $r\in R$ then $u=v$. Here and in what follows, $d(\cdot,\cdot)$ denotes the geodesic distance between pairs of vertices in the corresponding graph. $\beta(G)$, the metric dimension of $G$, is defined as the size of a smallest possible resolving set of $G$~\cite{Sla75,HarMel76}. The problem of finding the metric dimension of an arbitrary graph is NP-Complete~\cite{Coo71, GarJoh79, KhuRagRos96}. Nevertheless, when the distance matrix of a graph can be computed explicitly, resolving sets of size $\big(1+(1+o(1))\ln|V|\big)\cdot\beta(G)$ may be found using the so-called Information Content Heuristic (ICH)~\cite{HauSchVie12}. For a concise exposition of metric dimension see~\cite{TilFroLla19}, and for a detailed exposition see~\cite{TilFroLla21}.

An appealing aspect of resolving sets is their utility to represent nodes in graphs as Euclidean vectors---offering an alternative to other graph embedding techniques such as node2vec~\cite{GroLes16}. Indeed, if $R=\{r_1,\ldots,r_n\}$ of cardinality $n$ resolves $G$, then the transformation $d(v|R):=(d(v,r_1),\ldots,d(v,r_n))$, from $V$ into $\RR^n$, represents nodes in $G$ as $n$-dimensional vectors in a one-to-one manner. Further, $d(\cdot|R)$ maps nearby nodes in $G$ into tuples with similar coordinates in $\RR^n$. In particular, if the geodesic distance is of relevance for a node classification problem, resolving set based embeddings induce natural numerical features for the nodes in a graph~\cite{TilLla19}. Of course, the smaller the cardinality of a resolving set, the smaller the dimension of the associated Euclidean space, which motivates the study of metric dimension, and of algorithms capable of efficiently finding small resolving sets.

The Hamming distance between two strings $u$ and $v$ of the same length, denoted as $h(u,v)$, is the total number of mismatches between $u$ and $v$. (The length of a string $w$ is denoted $|w|$.) Up to a graph isomorphism, the Hamming graph $\HH_{k,a}$, with $k,a\ge1$ integers, has as vertices all strings of length $k$ formed using the characters in $\{0,\ldots,a-1\}$, and two vertices $u$ and $v$ are neighbors if and only if $h(u,v) = 1$. As a result, the geodesic distance between nodes in $\HH_{k,a}$ is precisely their Hamming distance; in particular, Hamming graphs are connected. We call $k$ the dimension and $a$ the alphabet size of $\HH_{k,a}$, respectively. 

Much is known already about Hamming graphs, including their automorphism group~\cite{ChaBer06} and their asymptotic metric dimension. Indeed~\cite{JiaPol19}:
\[\beta(\HH_{k,a})\sim\frac{2k}{\log_a(k)},\hbox{ as }k\to\infty,\]
and because the proof of this result is constructive, a resolving set of $\HH_{k,a}$ of approximate relative size $2k/\log_a(k)$ may be found for $k$ large enough. Otherwise, starting from a resolving set of $H_{k-r,a}$ of some size $s$ (e.g., obtained using the ICH), a resolving set for $\HH_{k,a}$ of size $s+r\lfloor a/2\rfloor$ may be found recursively in $O(ar^2)$ time~\cite{TilLla19}. Recent work has shown how to identify unnecessary nodes in a resolving set~\cite{Laietal19}; which may provide better non-asymptotic estimates for $\beta(\HH_{k,a})$.

As mentioned earlier, resolving sets of graphs are useful to represent their nodes as Euclidean vectors. In particular, resolving sets in Hamming graphs may be used to represent symbolic sequences (e.g., words and genomic sequences) numerically. Unfortunately, this capability is limited to sequences of the same length, and a chief motivation of this paper is to overcome this equal length limitation.

The Levenshtein distance~\cite{Lev66} (aka edit distance) between two strings $u$ and $v$ of possibly different lengths is defined as the minimal number of character substitutions, deletions, or insertions required to transform one string into the other. We denote this quantity as $\ell(u,v)$. Since the Hamming distance can be thought of as the minimal number of substitutions to transform one string into the other, if $|u|=|v|$ then $\ell(u,v)\le h(u,v)$. 

The Levenshtein distance can also be described as the least possible score (i.e. total number of mismatches, insertions or deletions) of an alignment between strings~\cite{DurEddKroMit98}. Traditionally, insertions and deletions are called ``indels,'' and denoted with the symbol $\mathtt{-}$. To fix ideas, equations (\ref{ex:align1})-(\ref{ex:align3}) display three alignments between the strings $001$ and $01$. The score of the alignment $A$ in (\ref{ex:align1}) is two because the second 0 in the first row is mismatched with the character 1 in the second row, and the 1 in the first row is aligned against an indel. Similarly, the scores of alignments $B$ and $C$ are one. Since the score of any alignment between different strings must be one or larger, $B$ and $C$ are optimal alignments and $\ell(001,01)=1$.
\begin{eqnarray}
A &:=& \begin{array}{cccc}
\mathtt{0} & \mathtt{0} & \mathtt{1}\\
\mathtt{0} & \mathtt{1} & \mathtt{-} 
\end{array}
\label{ex:align1}\\
B &:=& \begin{array}{cccc}
\mathtt{0} & \mathtt{0} & \mathtt{1}\\
\mathtt{0} & \mathtt{-} & \mathtt{1} 
\end{array}
\label{ex:align2}\\
C &:=& \begin{array}{cccc}
\mathtt{0} & \mathtt{0} & \mathtt{1}\\
\mathtt{-} & \mathtt{0} & \mathtt{1} 
\end{array}
\label{ex:align3}
\end{eqnarray}

Optimal alignments can be determined and scored through a well-known dynamic programming approach, which has been invented many times in different contexts~\cite{Lev66,NeeWun70,WagFis74}. For strings $u=u_1\ldots u_m$ and $v=v_1\ldots v_n$ of lengths $m$ and $n$, respectively, where $u_i$ and $v_j$ denote alphabet characters, this algorithm computes the columns (or rows) of the $m\times n$ matrix with entries $d_{i,j} := \ell (u_1\ldots u_i,v_1\ldots v_j)$ via the recursion:
\begin{equation}
d_{i,j} = \min\Big\{d_{i-1,j-1} + \indicator[u_i \neq v_j],d_{i-1,j} + 1,d_{i,j-1} + 1\Big\}.
\label{eq:dyn_prog}
\end{equation}
Here $\indicator[\cdot]$ is the indicator function of the proposition within. The time complexity of this algorithm is $O(mn)$, which is expensive for long pairs of strings; however, by focusing on the diagonals of the matrix $(d_{i,j})$, as oppose to its columns or rows, it is possible to speed up the calculations to an $O\big({\ell(u,v)}\cdot\min\{m,n\}\big)$ complexity~\cite{Ukk85}.

\subsection{Preliminaries and related work}

To overcome the length limitation of Hamming graphs, we adopt the following definition.

\begin{definition} 
For integers $0\le k_1\le k_2$ and $a\ge2$, the Levenshtein graph $\Lev_{k_1,k_2;a}$ has as vertices all strings of a length between $k_1$ and $k_2$ (inclusive) formed using the characters in $\{0,\ldots,a-1\}$, and two nodes $u$ and $v$ are connected by an edge iff $\ell(u,v)=1$. We denote the vertex and edge set of this graph as $V_{k_1,k_2;a}$ and $E_{k_1,k_2;a}$, respectively. (See Figure~\ref{fig:matrix_labels}.)
\label{def:levGr}
\end{definition}

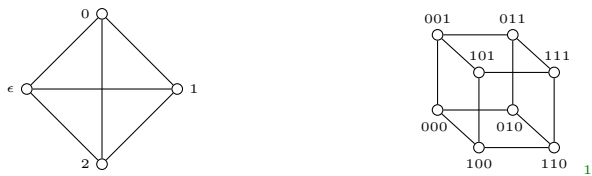
\begin{figure}[h]
    \centering
    \begin{tiny}
    \hfill
    \begin{subfigure}[b]{0.4\textwidth}
        \centering
        \begin{tikzpicture}[scale=1.0,
        node/.style={circle, draw, fill=blue!0,
        inner sep=0pt, minimum size = 4pt}]
            \node [node] (eps) at (0,1) [label=left:$\epsilon$] {};
            \node [node] (0) at (1,2) [label=left:$0$] {};
            \node [node] (1) at (2,1) [label=right:$1$] {};
            \node [node] (2) at (1,0) [label=left:$2$] {};
            \foreach \from/\to in {eps/0,eps/1,eps/2,0/1,0/2,1/2}
                { \draw (\from) -- (\to);} 
        \end{tikzpicture}
        \label{fig:Complete}
    \end{subfigure}
    \begin{subfigure}[b]{0.4\textwidth}
        \centering
        \begin{tikzpicture}[scale = 1.0,
        style/.style={circle, draw, fill=blue!0,
        inner sep=0pt, minimum size = 4pt}]
            \node [style] (000) at (0,1) [label=below:$000\,\,$] {};
            \node [style] (001) at (0,2) [label=above:$001$] {};
            \node [style] (010) at (1,1) [label=below:$010\,\,$] {};
            \node [style] (011) at (1,2) [label=above:$011$] {};
            \node [style] (100) at (0.55,0.5) [label=below:$100$] {};
            \node [style] (101) at (0.55,1.5) [label=above:$\,\,101$] {};
            \node [style] (110) at (1.55,0.5) [label=below:$110$] {};
            \node [style] (111) at (1.55,1.5) [label=above:$\,\,111$] {};
            \foreach \from/\to in {000/001,000/010,000/100,001/011,001/101,010/011,010/110,100/101,100/110,011/111,101/111,110/111}
                { \draw (\from) -- (\to);} 
        \end{tikzpicture}
      \ref{fig:matrix_labels}
    \end{subfigure}
    \hfill
    \end{tiny}
    \caption{Visual representation of $\mathbb{L}_{0,1;3}$ (left), and  $\mathbb{L}_{3,3;2}$ (right).}
    \label{fig:matrix_labels}
\end{figure}

Observe that, for $k_1\le k\le k_2$, the subgraph of nodes in $\Lev_{k_1,k_2;a}$ of length $k$ is precisely $\HH_{k,a}$. Further, only nodes of equal or consecutive length can be neighbors in $\Lev_{k_1,k_2;a}$ (see Figure~\ref{fig:Lev_Graph_Ex}).

Ahead we write $\mathbb{L}_{k;a}$ as shorthand for $\mathbb{L}_{0,k;a}$. Accordingly, we denote the vertex and edge set of $\mathbb{L}_{k;a}$ as $V_{k;a}$ and $E_{k;a}$, respectively. The empty string, denoted as $\epsilon$, is the only vertex of length zero in this graph. Besides, we define $\mathbb{L}_a$ as the graph with vertex set $\cup_{k\ge1}V_{k;a}$ where two nodes $u$ and $v$ of arbitrary length are neighbors if and only if $\ell(u,v)=1$. All nodes in $\mathbb{L}_a$ have finite length.

Various other notions of Levenshtein graphs have been considered in the literature, usually motivated by specific applications. One common definition is that two nodes are neighbors when their Levenshtein distance is underneath some threshold. For instance, Pisanti, Et, and Diderot~\cite{PisEtDid98} define Levenshtein graphs over a vertex set of arbitrary genes, and two genes $u$ and $v$ are joined by an edge when $\ell(u,v)\le t$; which they use to test random graphs as viable models for genomic data. Instead, Sala et al.~\cite{Saletal15} define the vertex set of Levenshtein graphs as $\{0,\ldots,a-1\}^k$, and $u$ and $v$ are neighbors only when $\ell(u,v)\le 2t$; they use this to help expand on information about the number of common subsequences and supersequences a pair of strings have. Zhong, Heinicke, and Rayner~\cite{ZhoHeiRay19} define the vertex set of the Levenshtein graph to have nodes corresponding to microRNAs in mice and people, and $u$ and $v$ are connected by an edge only when $\ell(u,v)\leq 3$. Finally, Stahlberg~\cite{Sta11} defines the vertex set of Levenshtein graphs from all strings of a given set $M$ as well as all strings that lie on a shortest path between two strings in $M$, and nodes $u$ and $v$ are then joined by an edge if and only if $\ell(u,v)=1$. 

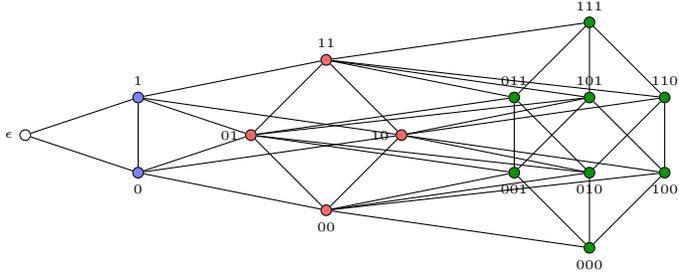
\begin{figure}
\centering
    \begin{tiny}
    \begin{tikzpicture}[scale=1.0,
    node_0/.style={circle, draw, fill=blue!0,
        inner sep=0pt, minimum size = 4pt},
    node_1/.style={circle, draw, fill=blue!50,
        inner sep=0pt, minimum size = 4pt},
    node_2/.style={circle, draw, fill=red!60,
        inner sep=0pt, minimum size = 4pt},
    node_3/.style={circle, draw, fill=black!40!green,
        inner sep=0pt, minimum size = 4pt}]
        
        \node [node_0] (e) at (0,1) [label=left:$\epsilon$] {};
        
        \node [node_1] (0) at (1.5,.5) [label=below:$0$] {};
        \node [node_1] (1) at (1.5,1.5) [label=above:$1$] {};
        
        \node [node_2] (00) at (4,0) [label=below:$00$] {};
        \node [node_2] (11) at (4,2) [label=above:$11$] {};
        \node [node_2] (01) at (3,1) [label=left:$01$] {};
        \node [node_2] (10) at (5,1) [label=left:$10$] {};
        
        \node [node_3] (000) at (7.5,-0.5) [label=below:$000$] {};
        \node [node_3] (001) at (6.5,.5)   [label=below:$001$] {};
        \node [node_3] (010) at (7.5,.5)   [label=below:$010$] {};
        \node [node_3] (100) at (8.5,.5)   [label=below:$100$] {};
        \node [node_3] (011) at (6.5,1.5)  [label=above:$011$] {};
        \node [node_3] (101) at (7.5,1.5)  [label=above:$101$] {};
        \node [node_3] (110) at (8.5,1.5)  [label=above:$110$] {};
        \node [node_3] (111) at (7.5,2.5)  [label=above:$111$] {};
            
        \foreach [count=\i] \x/\y in {e/0, e/1, 0/1,
            0/00, 0/01, 0/10, 1/01, 1/10, 1/11,
            00/01, 00/10, 01/11, 10/11,
            00/000, 00/001, 00/010, 00/100,
            01/001, 01/010, 01/011, 01/101,
            10/010, 10/100, 10/101, 10/110,
            11/011, 11/101, 11/110, 11/111,
            000/001, 000/010, 000/100,
            001/011, 001/101, 010/011, 010/110, 100/101, 100/110,
            011/111, 101/111, 110/111}
                \draw (\x) -- (\y);
    \end{tikzpicture}
    \end{tiny}
    \caption{Visual representation of $\mathbb{L}_{3;2}$. The sub-graphs of all strings of fixed length are Hamming graphs: the white, blue, red, and green nodes form $\HH_{0,2}$, $\HH_{1,2}$, $\HH_{2,2}$, and $\HH_{3,2}$, respectively.}
    \label{fig:Lev_Graph_Ex}
\end{figure}

Since $\Lev_{k,k;a}$ is isomorphic to $\HH_{k,a}$; Levenshtein graphs include Hamming graphs as special cases. Nevertheless, as pointed out in~\cite{VarKusGoy16}, which implicitly uses a notion similar to ours, Levenshtein graphs cannot be represented as Cartesian products when $k_1<k_2$. This makes their study particularly challenging.

In this manuscript we initiate a study of Levenshestein graphs---as given in Definition~\ref{def:levGr}. The manuscript is based on the recent Honors Thesis by the first author~\cite{Rut21}. 

\subsection{Paper organization}

In Section \ref{sec:LevGr}, we show that Levenshtein graphs are always connected, and provide a necessary and sufficient condition for the geodesic distance to coincide with the edit distance between pairs of nodes. Unlike Hamming graphs, the edit and geodesic distance between all pairs of nodes in a Levenshtein graph is not necessarily the same. For instance, in $\Lev_{3,3;2}$, $\ell(010,101)=2$ but $d(010,101)=3$ (see Figure~\ref{fig:matrix_labels}). Nevertheless, in $\Lev_{0,3;2}$, $d(010,101)=2$ (see Figure~\ref{fig:Lev_Graph_Ex}).

In Section \ref{sec:runRes}, we show a formula to describe the edit distance of an arbitrary string to a string with at most two runs (a run is a maximal substring of a single repeated character in a string). This formula leads to an algorithm to compute the distance from any string $u$ to any string with at most two runs in $O(|u|)$ time, which is faster than many common methods of computing the edit distance.  The results in sections~\ref{sec:metDim}-\ref{sec:auto} rely heavily on Section \ref{sec:runRes}. In Section~\ref{sec:metDim}, we construct a resolving of $\Lev_{k_1,k_2;a}$ of size $O\left(a k_2(k_2-k_1+1)\right)$ explicitly. Since nodes on this set have at most two runs, we may utilize the algorithm from Section \ref{sec:runRes} to multilaterate efficiently any string of length between $k_1$ and $k_2$.

In Section \ref{sec:auto}, we characterize the automorphism group of Levenshtein graphs, which has fixed size $2a!$ when $k_1<k_2$ and $k_2\ge2$. Finally, in Section \ref{sec:detSet}, we address the determining number of Levenshtein graphs. This notion is useful for describing graph automorphisms. For a given graph $G=(V,E)$, a set $S\subset V$ is called determining if whenever $f$ and $g$ are automorphisms of $G$ such that $f(s)=g(s)$, for all $s\in S$, then $f=g$. The determining number of a graph is the size of its smallest determining set. For $k_1<k_2$ with $k_2\ge2$ and $(k_2,a)\ne(2,2)$, we show that the determining number of $\Lev_{k_1,k_2;a}$ is $\lceil a/k_2\rceil$.

\section{Geodesic versus Edit Distance, and Connectivity}
\label{sec:LevGr}

The geodesic distance between pairs of nodes in a Hamming graph is equal to their Hamming distance; however, as already pointed out in the Introduction, this is not necessarily the case for Levenshtein graphs. The main result in this section is the following one.

\begin{theorem}
Levenshtein graphs are connected, and the geodesic distance between every pair of nodes on $\Lev_{k_1,k_2;a}$ is equal to their Levenshtein distance if and only if $k_1<k_2$ or $k_1 = k_2\leq 2$. If $k>2$ then the geodesic distance in $\Lev_{k,k;a}$ is the Hamming distance.
\label{thm:geo_dist}
\end{theorem}

This theorem is a direct consequence of the following three lemmas. 

Ahead, the length of a path is understood as the number edges that compose it. In addition, $w_{(n)}$ and $w^{(n)}$ denote the prefix and suffix of length $n$ of a word $w$, respectively.

\begin{lemma}
Let $k_1 < k_2$. For all nodes $u$ and $v$ in $\Lev_{k_1,k_2;a}$, there is a path of length $\ell(u,v)$ that connects $u$ with $v$. In particular, $\Lev_{k_1,k_2;a}$ is connected, and for all $u,v\in V_{k_1,k_2;a}$, $d(u,v)\le \ell(u,v)$.
\label{lem:La}
\end{lemma}
\begin{proof}
We show something more general, namely, for any alignment between two nodes $\Lev_{k_1,k_2;a}$, there is a path of the same length as the alignment score that connects them, while visiting only nodes of a length between the shortest and longest of the two.

Consider a fixed alignment $A$ between two nodes $u$ and $v$. Define $\delta:=|u|-|v|$. Since  alignment scores are invariant under permutations of their rows, as well as their columns, we may assume without any loss of generality that $|u|\ge|v|$, and that $A$ is of the form:
\[A=\left.\begin{array}{c} u_0\\ v_0\end{array}\right|
\begin{array}{c} u_1\\ \mathtt{-}^\delta\end{array}\left|
\begin{array}{c} u_2\\ \mathtt{-}^k\end{array}\right|
\begin{array}{c} \mathtt{-}^k\\
v_2\end{array};
\]
where the $u_i$'s and $v_i$'s are nodes in $\Lev_{k_1,k_2;a}$ such that $|u_0|=|v_0|\ge0$, $|u_1|=\delta$, $|u_2|=|v_2|=k$ for some $k\ge0$, and $\mathtt{-}^n$ denotes $n$ consecutive gaps.

Let $s_0$ denote the score of the alignment associated with $u_0$ and $v_0$ above. Clearly, we can construct a path of length $s_0$ from $u=u_0u_1u_2$ to $v_0u_1u_2$ substituting, one at a time, the mismatched characters in $u_0$ by the corresponding characters in $v_0$. Since substitutions do not alter the length of a node, all nodes in this path have length $|u|$.

Next, we can construct a path of length $\delta$ from $v_0u_1u_2$ to $v_0u_2$ deleting, one at a time, the characters in $u_1$. In particular, the nodes in this path have a (decreasing) length between $|v_0u_1u_2|=|u|$ and $|v_0u_2|=|v|$, inclusive.

We can now construct a path of length $2k$ from $v_0u_2$ to $v_0v_2=v$, stitching the following paths of length 2. When $|v|<k_2$, each of these paths is obtained by inserting a character from $v_2$, and subsequently deleting another in $u_2$. As a result, all nodes in these paths have a length between $|v|$ and $|v|+1\le k_2$, inclusive. The short paths are:
\begin{align*}
& v_0\,u_2^{(k)}\, {v_2}_{(0)},\,
v_0\,u_2^{(k-1)}\,{v_2}_{(0)},\,
v_0\,u_2^{(k-1)}\,{v_2}_{(1)}; \\
& v_0\,u_2^{(k-1)}\,{v_2}_{(1)},\,
v_0\,u_2^{(k-2)}\,{v_2}_{(1)},\,
v_0\,u_2^{(k-2)}\,{v_2}_{(2)};\\
& \vdots \\
& v_0\,u_2^{(1)}\,{v_2}_{(k-1)},\,
v_0\,u_2^{(0)}\,{v_2}_{(k-1)},\,
v_0\,u_2^{(0)}\,{v_2}_{(k)}.
\end{align*}
Similarly, when $|v| = k_2$, each of these paths is obtained by deleting a character in $v_2$, and subsequently inserting a character from $u_2$. All nodes in these paths have a length between $|v|$ and $|v|-1 \ge k_1$ inclusive.

Appending all the previous paths, we obtain a path from $u$ to $v$ of length $s_0+\delta+2k$, which is precisely the score of $A$. This shows the lemma because each node in this path is contained in $\Lev_{k_1,k_2;a}$.
\end{proof}

\begin{lemma}
Let $k_1<k_2$. For all nodes $u$ and $v$ in $\Lev_{k_1,k_2;a}$, $d(u,v) \geq\ell(u,v)$.
\label{lem:GeoGeqLev}
\end{lemma}

\begin{proof}
Clearly, $d(u,v)=0$ if and only if $\ell(u,v)=0$. Thus, without loss of generality, we may assume that $n:=d(u,v)\ge1$. Due to Lemma~\ref{lem:La}, $n$ is finite; in particular, there is in $\Lev_{k_1,k_2;a}$ a (simple) path $w_0=u,\ldots,w_n=v$ of length $n$ that connects $u$ and $v$. Since $d(w_i,w_{i+1})=\ell(w_i,w_{i+1})=1$, the triangular inequality implies that:
\begin{align*}
d(u,v)=\sum_{i=0}^{n-1}d(w_i,w_{i+1})=\sum_{i=0}^{n-1}\ell(w_i,w_{i+1})\ge\ell(u,v),
\end{align*}
which shows the lemma.
\end{proof}

\begin{lemma}
For all $k\ge0$, $\Lev_{k,k;a}=\mathbb{H}_{k;a}$; in particular, $\Lev_{k,k;a}$ is connected. Further, the geodesic distance between every pair of nodes on $\Lev_{k,k;a}$ is equal to their Levenshtein distance if and only if $k\leq 2$.
\label{lem:geokka}
\end{lemma}
\begin{proof}
To show the first claim, it suffices to show that $\Lev_{k,k;a}$ and $\mathbb{H}_{k,a}$ have the same edges. Indeed, if $h(u,v)=1$ then $u$ and $v$ can be aligned perfectly except for one mismatch. In particular, $\ell(u,v)\le1$. But, since $u\ne v$, $\ell(u,v)>0$, hence $\ell(u,v)=1$. Conversely, if $\ell(u,v)=1$ then an optimal alignment between $u$ and $v$ consists of a single mismatch, or a single indel. Since the latter is not possible because $|u|=|v|$, $h(u,v)=1$, which shows the claim.

Due to the first claim, $d(u,v)=h(u,v)$ for all pair of nodes $u,v$ in $\Lev_{k,k;a}$. We use this to show the second claim, assuming, without loss of generality, that $u\ne v$.

The second claim is trivial when $k=0$. If $k=1$ then, as we argued before, $\ell(u,v)=1=h(u,v)=d(u,v)$. Instead, if $k=2$ and $h(u,v)=1$ then, as we just argued, $\ell(u,v)=1=h(u,v)=d(u,v)$. Otherwise, if $k=2$ but $h(u,v)=2$ then Lemma~\ref{lem:GeoGeqLev} implies that $0<\ell(u,v)\le2$; however, $\ell(u,v)=1$ is not possible because the optimal alignment between $u$ and $v$ would then have to use a single indel, which in turn is not possible because $u$ and $v$ are of the same length. Hence, $\ell(u,v)=2$ and again $\ell(u,v)=h(u,v)=d(u,v)$. 

Finally, if $k>2$, and since $a\ge2$, there is in $\Lev_{k,k;a}$ a node $u$ of length $k$ formed by alternating $0$'s and $1$'s. Let $v$ be the flip of $u$. Then $h(u,v)=k$ but $\ell(u,v)\le2$ because the strings $\mathtt{-}u$ and $v\mathtt{-}$ align perfectly except for their ends; in particular, $h(u,v)>\ell(u,v)$ i.e. $d(u,v)>\ell(u,v)$.    
\end{proof}

\section{Levenshtein distance to a string with at most two runs}
\label{sec:runRes}

In this section, we obtain rather explicit formulas for the edit distance between an arbitrary string and another one with at most two runs. These will prove useful for studying the resolvability of Levenshtein graphs and their automorphism group.

In what follows the total number of occurrences of an alphabet character $\alpha$ in a string $w$ is denoted $N_\alpha(w)$, whereas the number of runs in $w$ is denoted $r(w)$. For example, $N_0(01121)=1$, $N_1(01121)=3$, $N_2(01121)=1$, and $r(01121)=4$.

The main result in this section is the following.

\begin{theorem}
Let $l,r\ge0$ be integers and $\alpha,\beta$ different alphabet characters. Then, for all string $w$:
\begin{align}
\label{ide:L21run} \ell(w,\alpha^l) &= \max\{\abs w,l\} - \min\{N_\alpha(w),l\};\\
\label{ide:L22run} \ell(w,\alpha^l\beta^r) &= \min_{i_0\le i \le i_1} \ell\big(w_{(i)},\alpha^l\big) + \ell\big(w^{(|w|-i)},\beta^r\big);
\end{align}
where $i_0 := \max\{0,\min\{l,|w|-r\}\}$ and $i_1 := \min\{|w|,\max\{l,|w|-r\}\}$.
\label{thm:masterruns}
\end{theorem}

A noteworthy consequence of this theorem is the following.

\begin{corollary}
If $u$ and $v$ are strings such that $|u|=|v|$, and $u$ or $v$ have at most two runs, then $\ell(u,v)=h(u,v)$.
\label{cor:Lev-Ham}
\end{corollary}
\begin{proof}
Suppose that $|u|=|v|=k$, and write $u=u_1\cdots u_k$ with $u_1,\ldots,u_k$ alphabet characters. Without any loss of generality assume that $r(v)\le2$.

If $r(v)=0$ then $u=v$; in particular, $\ell(u,v)=0=h(u,v)$. Instead, if $r(v)=1$ then $v=\alpha^k$ for some alphabet character $\alpha$, and Equation~(\ref{ide:L21run}) implies that \[\ell(u,v)=k-N_\alpha(u)=\sum_{i=1}^k\indicator[u_i\ne\alpha]=h(u,v).\]
Finally, if $r(v)=2$ then $v=\alpha^l\beta^{k-l}$ for some integer $1\le l<k$ and alphabet characters $\alpha\ne\beta$. Hence, from Equation~(\ref{ide:L21run}), and the previous result for when $r(v)=1$, we find that
\begin{align*}
\ell(u,v) &=\ell(u_1\cdots u_l,\alpha^l)+\ell(u_{l+1}\cdots u_k,\beta^{k-l})\\
&=h(u_1\cdots u_l,\alpha^l)+h(u_{l+1}\cdots u_k,\beta^{k-l})=h(u,v),    
\end{align*}
as claimed.
\end{proof}

The proof of Theorem~\ref{thm:masterruns} follows from the next two results. Equation~(\ref{ide:L21run}) is a direct consequence of Lemma~\ref{lem:oneRun}, and equation~(\ref{ide:L22run}) follows from Lemma~\ref{lem:two_runs}.

\begin{lemma}
For all string $w$, if $l\ge 0$ and $\alpha$ is an alphabet character then: $\ell(w,\alpha^l) = \max\{\abs w,l\} - \min\{N_\alpha(w),l\}.$ 
\label{lem:oneRun}
\end{lemma}

\begin{proof}
Assume that $w\ne\epsilon$ and $l>0$, otherwise the statement is trivial. The score of an alignment is its length minus the number of matches in it. But the length of an alignment is at least the length of the longest string, and the number of matches is at most the number of characters shared by the strings. In particular, since the edit distance between $w$ and $\alpha^l$ is the score of some optimal alignment, we have that:  $\ell(w,\alpha^l) \ge \max\{\abs w,l\}-\min\{N_\alpha(w),l\}$. 

To complete the proof, it suffices to expose an alignment with the same score as the right-hand side of this inequality. For this let $n:=N_\alpha(\omega)$. Assume first that $\alpha^n$ is a prefix of $w$. We now consider two cases. If $|w|\le l$ then $w=\alpha^nu$, with $N_\alpha(u)=0$, and the following alignment between $w$ and $\alpha^l$ has the desired score:
\[\left.\begin{array}{c}
\alpha^n \\
\alpha^n
\end{array}\right|
\left.\begin{array}{c}
u \\
\alpha^{|w|-n}
\end{array}\right|
\left.\begin{array}{c}
\mathtt{-}^{l-|w|} \\
\alpha^{l-|w|}
\end{array}\right|.\]
Otherwise, if $|w|\ge l$, let $\delta=\min\{n,l\}$ and write $w=\alpha^\delta u v$, with $|u|=l-\delta$ and $|v|=|w|-l$. Now, the following alignment has the desired score:
\[
\left.\begin{array}{c}
\alpha^\delta \\
\alpha^\delta
\end{array}\right|
\left.\begin{array}{c}
u \\
\alpha^{l-\delta}
\end{array}\right|
\begin{array}{c}
v \\
\mathtt{-}^{|w|-l}
\end{array}.\]
The previous argument assumes that $\alpha^n$ is a prefix of $w$. If this is not the case, we may shuffle the columns of the alignments to reproduce $w$ on the top row but without altering their scores. From this, the lemma follows.
\end{proof}

\begin{lemma}
Let $k,l,r\ge0$ be integers. If $w = w_1\cdots w_k$ is a string of length $k$ and $\alpha$, $\beta$ are different alphabet characters then
\[\ell(w,\alpha^l\beta^r) = \min_{i_0\le i \le i_1} \ell(w_{(i)},\alpha^l) + \ell(w^{(k-i)},\beta^r),\]
where $i_0 := \max\{0,\min\{l,k-r\}\}$ and $i_1 := \min\{k,\max\{l,k-r\}\}$.
\label{lem:two_runs}
\end{lemma}

\begin{proof}
Without loss of generality assume that $k>0$. Define $l_i:=N_\alpha(w_{(i)})$ and $r_i:=N_\beta(w^{(k-i)})$, for $0<i<k$. Furthermore, define $l_i:=0$ and $r_i:=N_\beta(w)$ for $i\le0$, and $l_i:=N_\alpha(w)$ and $r_i:=0$ for $i\ge k$.

Any alignment $A$ between $w$ and $\alpha^l\beta^r$ may be segmented as
\[A=\left.
\begin{array}{c}
u_0 \\
v_0
\end{array}\right|
\begin{array}{c}
u_1 \\
v_1
\end{array},\]
where $u_0$ and $u_1$ correspond to a possibly empty prefix and suffix of $w$, respectively, and $v_0$ and $v_1$ correspond to the strings $\alpha^l$ and $\beta^r$, respectively. ($u_0,u_1,v_0,v_1$ may contain $\mathtt{-}$'s.) Since this also applies to an optimal alignment between $w$ and $\alpha^l\beta^r$, it follows that
\begin{align*}
\ell(w,\alpha^l\beta^r) 
&= \min_{0\le i\le k}\quad\ell(w_{(i)},\alpha^l)+\ell(w^{(k-i)},\beta^r)\\
&= \min_{0\le i\le k}\quad \max\{l,i\}-\min\{l,l_i\}+\max\{r,k-i\}-\min\{r,r_i\}\\
&= \min_{0\le i\le k}\quad\frac{k+|l-i| + |k-r-i| + |l-l_i| - l_i + |r-r_i|  - r_i}{2},
\end{align*}
where for the second identity we have used Lemma~\ref{lem:oneRun}, and for the third one the well-known identities $\max\{a,b\}=(a+b+|a-b|)/2$, and $\min\{a,b\}=(a+b-|a-b|)/2$.

Consider the functions $f_1,f_2:\mathbb{Z}\to\mathbb{Z}$ defined as
\begin{align*}
f_1(i) &:= \frac{k-l-r}{2}+\frac{|l-i|+|k-r-i|}{2}\\
f_2(i) &:= \frac{|l-l_i|+l-l_i}{2}+\frac{|r-r_i|+r- r_i}{2}.
\end{align*}
In particular, $\ell(w,\alpha^l\beta^r)=\min_{0\le i\le k}f_1(i)+f_2(i)$. Next we show that this minimum is achieved at some $i_0\le i\le i_1$.

Observe that up to a constant summand, $f_1(i)$ is the average of the distance from $i$ to $l$, and from $i$ to $k-r$. So $f_1(i)$ is strictly decreasing for $i\le\min\{i,k-r\}$, and strictly increasing for $\max\{i,k-r\}\le i$. In particular, when restricted to the domain $\{0,\ldots,k\}$, $f_1$ is monotone decreasing to the left of $i_0$, constant between $i_0$ and $i_1$, and monotone increasing to the right of $i_1$. Note that $f_1(i)=|u|-l-r$, for $i_0\le i \le i_1$.

On the other hand, observe that $f_2(i)=g(l-l_i)+g(r-r_i)$, where 
\[g(x):=\frac{|x|+x}{2},\hbox{ for }x\in\mathbb{Z};\]
satisfies $|g(x)-g(x-1)|\le1$. In particular, if $w_{i+1}=\alpha$ then $|f_2(i+1)-f_2(i)|\le1$ because $l_{i+1}=l_i+1$ and $r_{i+1}=r_i$. Similarly, if $w_{i+1}=\beta$ then $|f_2(i+1)-f_2(i)|\le1$ because $l_{i+1}=l_i$ and $r_{i+1}=r_i-1$. Finally, if $w_{i+1}\notin\{\alpha,\beta\}$ then $l_{i+1}=l_i$ and $r_{i+1}=r_i$, hence $f_2(i+1)=f_2(i)$. In either case, we find that $|f_2(i+1)-f_2(i)|\le1$ for $0\le i<k$. As a result, since $f_1$ is integer-valued, $f_1+f_2$ is decreasing for $i\le i_0$ but increasing for $i_1\le i$, from which the lemma follows.
\end{proof}

\subsection{Efficient algorithmic calculation}

The proof of Lemma \ref{lem:two_runs} can be ad\-apted into a method (see Algorithm~\ref{alg:tworun}) that finds the distance between an arbitrary string $w$ to a string of the form $v=\alpha^l\beta^r$ in $O(|w|)$ time---assuming that $\alpha,\beta,l,$ and $r$ are known in advance. The algorithm exploits that $f_1(i)$ is constant for $i_0\le i \le i_1$, reducing the calculation of $\ell(w,v)$ to minimizing $f_2$ over the restricted domain. This can be done through a loop where $f_2(i_0)$ can be found directly, and the remaining values can be found recursively by finding $f_2(i+1)-f_2(i)$ through cases depending on $l_i$,  $r_i$, and $w_{i+1}$. This is faster than standard methods of finding the edit distance between strings with $O(|w||v|))$ time complexity .
	
A number of papers suggest methods for effectively computing the edit distance between run-length encoded strings~\cite{ArbLanMit02,MakUkkNav03}. These methods adapt the standard dynamic programming approach to compute $\ell(u,v)$ in $O(r(u)|v|+r(v)|u|)$ time. Comparatively, Algorithm~\ref{alg:tworun} has a few benefits and quirks: it assumes only one string is run-length encoded, it is fast due to specificity, and it provides a formula that is useful for proofs.

\begin{algorithm}
\caption{for computing the edit distance to a two-run string}
\label{alg:tworun}
\begin{algorithmic}
      \STATE \textbf{Input.} $w$ a string, $\alpha\ne\beta$ alphabet characters, and $l,r>0$ integers
      \STATE  \textbf{Output.} $\ell(w,\alpha^l\beta^r)$
      \STATE $k\gets|w|$
      \STATE $i_0 \gets \max\{0,\min\{l,k-r\}\}$
      \STATE $i_1 \gets \min\{k,\max\{l,k-r\}\}$
      \STATE $l_i \gets N_\alpha(w_1\cdots w_{i_0})$ 
      \STATE $r_i \gets N_\beta(w_{i_0+1}\cdots w_k)$
      \STATE $f_2 \gets (|l-l_i|+l-l_i)/2+(|r-r_i|+r-r_i)/2$
      \STATE $m\gets f_2$
      
      \FOR{$i = i_0+1$ to $i_1$}
        \IF{$w_i = \beta$}
            \IF{$r_i\le r$}
                \STATE $f_2 \gets f_2+1$
            \ENDIF
            \STATE $r_i \gets r_i-1$
        \ENDIF
      
        \IF{$w_i = \alpha$ and $l_i < l$}
            \STATE $f_2 \gets f_2-1$
            \STATE $m\gets \min\{m,f_2\}$ 
            \STATE $l_i \gets l_i + 1$
        \ENDIF
      \ENDFOR
    
      \STATE $f_1 \gets (k-l-r)/2+(|k-l-r|)/2$
      \RETURN{$f_1+m$}
\end{algorithmic}
\end{algorithm}

\section{Metric Dimension of Levenshtein Graphs}
\label{sec:metDim}

Recall that a subset of nodes $R$ in a graph $G$ is said to resolve it when $R$ resolves all pairs of different nodes, namely, for all nodes $u$ and $v$, with $u\ne v$, there exists $r\in R$ such that $d(u,r)\ne d(v,r)$. The metric dimension of the graph, $\beta(G)$, is the size of its smallest resolving set. 

The main result in this section are the following bounds on the metric dimension of Levenshtein graphs.

\begin{theorem} For all $0\le k_1\le k_2$ and $a\ge0$:
\begin{equation*}
O\left(\frac{k_2}{\log_a k_2}\right)\le\beta(\Lev_{k_1,k_2;a})\le O\left(a\big((k_2+1)^2-k_1^2\big)\right).
\end{equation*}
In particular, if $\Delta := k_2-k_1+1$ then $\beta(\Lev_{k_1,k_2;a})=O(ak_2\Delta)$.
\label{thm:betaLev}
\end{theorem}

Observe that if $\Delta=\Theta(k_2)$ then $\beta(\Lev_{k_1,k_2;a})$ grows at most quadratically in terms of the maximum string length $k_2$. However, if $\Delta=\Theta(1)$ then $\beta(\Lev_{k_1,k_2;a})$ grows linearly with the largest string length. By setting $k_1=k_2$, Theorem~\ref{thm:betaLev} may be applied to Hamming graphs as well. In this case, the lower bound of the Corollary is within a factor of 2 of the true asymptotic value.

The remaining of this section is devoted to proving Theorem~\ref{thm:betaLev}. The lower-bound is almost immediate from the following general inequality~\cite[Theorem 3.6]{KhuRagRos96}: if $G=(V,E)$ is a graph with metric dimension $\beta$ and diameter $\delta$ then $|V| \le \delta^\beta + \beta$. Observe that the diameter of $\Lev_{k_1,k_2;a}$ is at most $k_2$ because $\ell(u,v)\le\max\{|u|,|v|\}$, for all pair of strings $u$ and $v$. So, if $\beta=\beta(\Lev_{k_1,k_2;a})$ then
\begin{equation*}
a^{k_2}\le|V_{k_1,k_2;a}| \le k_2 ^ \beta + \beta \le (k_2 + 1) ^ \beta,
\end{equation*}
from which the left-hand side inequality in Theorem~\ref{thm:betaLev} follows. (In the above argument the inequality $|V_{k_1,k_2;a}|\ge a^{k_2}$, which neglects the parameter $k_1$, may seem absurdly loose; however, this is not the case because $|V_{k_1,k_2;a}|\le2a^{k_2}$.)

The upper-bound in Theorem~\ref{thm:betaLev} follows directly from the following three results.

\begin{lemma}
Let $k_1\le k\le k_2$. In $\Lev_{k_1,k_2;a}$, the following subset of nodes resolves any pair of different strings of length $k$:
\begin{align}
R_{k,a} := \bigcup\limits_{n=0}^{\lfloor a/2 \rfloor-1} \Big\{(2n)^i(2n+1)^{k-i}:0\le i\le k\Big\}. 
\label{eq:res_Set_Fix_k}
\end{align}
\label{lem:res}
\end{lemma}
\begin{proof}
Let $u=u_1\cdots u_k$ and $v=v_1\cdots v_k$ be nodes in $\Lev_{k_1,k_2;a}$ of the same length $k$ that differ at certain position $j$. Define $\alpha:=u_j$. Without loss of generality assume that $\alpha\ne(a-1)$ when $a$ is odd.

Due to Theorem~\ref{thm:geo_dist}, the geodesic distance between pairs of nodes in $\Lev_{k_1,k_2;a}$ is either their Hamming or Levenshtein distance. But, since nodes in $R_{k,a}$ have at most two runs, Corollary \ref{cor:Lev-Ham} implies that $\ell(u,r)=h(u,r)$ and $\ell(v,r)=h(v,r)$, for each $r\in R$. Hence, the geodesic distance between $u$ and $v$ to any node in $R_{k,a}$ is always the Hamming distance.

If $\alpha$ is even, we claim that $\{\alpha^{j-1}(\alpha\!+\!1)^{k-j+1},\alpha^j(\alpha\!+\!1)^{k-j}\}$ resolves $u$ and $v$. By contradiction suppose otherwise, i.e. assume that $d(u,\alpha^{j-1}(\alpha\!+\!1)^{k-j+1})= d(v,\alpha^{j-1}(\alpha\!+\!1)^{k-j+1})$ and $d(u,\alpha^j(\alpha\!+\!1)^{k-j})= d(v,\alpha^j(\alpha\!+\!1)^{k-j})$. If $\delta$ is the geodesic distance between $u$ (or $v$) and $\alpha^{j-1}(\alpha\!+\!1)^{k-j+1}$ then
\begin{align*}
d(u,\alpha^j(\alpha\!+\!1)^{k-j})
&= h(u,\alpha^j(\alpha\!+\!1)^{k-j})\\
&=\sum_{i=1}^{j-1} \indicator[u_i \ne \alpha] +\indicator[u_j\ne \alpha]+ \sum_{i=j+1}^k \indicator[u_i \ne \alpha+1]\pm\indicator[u_j\ne \alpha+1]\\
&= h(u,\alpha^{j-1}(\alpha\!+\!1)^{k-j+1})-1\\
&= \delta-1.
\end{align*}
On the other hand, since $v_j\ne\alpha$:
\begin{align*}
d(v,\alpha^j(\alpha\!+\!1)^{k-j})
&= h(v,\alpha^j(\alpha\!+\!1)^{k-j})\\
&= \sum_{i=1}^{j-1} \indicator[v_i \ne \alpha] + \indicator[v_j\ne \alpha]+\sum_{i=j+1}^k \indicator[v_i \ne \alpha+1]\pm\indicator[v_j\ne \alpha+1]\\
&=h(v,\alpha^{j-1}(\alpha\!+\!1)^{k-j+1})+1-\indicator[v_j\ne \alpha+1]\\
&\ge \delta,
\end{align*}
implying that $d(u,\alpha^j(\alpha\!+\!1)^{k-j})\ne d(v,\alpha^j(\alpha\!+\!1)^{k-j})$, which is not possible. So, $\{\alpha^{j-1}(\alpha\!+\!1)^{k-j+1},\alpha^j(\alpha\!+\!1)^{k-j}\}$ resolves $u$ and $v$.

Likewise, if $\alpha$ is odd, one can show that $\{(\alpha-1)^{i-1}\alpha^{k-i+1},(\alpha-1)^i\alpha^{k-i}\}$ resolves $u$ and $v$, from which the lemma follows.
\end{proof}

\begin{lemma}
If $\theta$ is the string bijection induced by the transformation $\theta(\alpha) := (\alpha+1) \pmod a$, for $\alpha\in\{0,\ldots,a-1\}$, then the set $\theta(R_{k-1;a})\cup R_{k+1;a}$ resolves all pairs of different strings of length $k$ that are permutations of each other.
\label{lem:const_red}
\end{lemma}
\begin{proof}
Recall that $w_{(n)}$ and $w^{(n)}$ denote the prefix and suffix of a string $w$ of length $n$, respectively.

Let $u$ be a string of length $k>1$, and $v\ne u$ correspond to a permutation of the characters in $u$. Let $i$ be the first position at which $u$ and $v$ differ; in particular, $u_{(i-1)} = v_{(i-1)}$, and $u^{(k-i+1)}$ and $v^{(k-i+1)}$ are permutations of each other.
We show the lemma by cases, see Figure~\ref{fig:VennD}.

\begin{figure}
\centering
\includegraphics[scale=0.5]{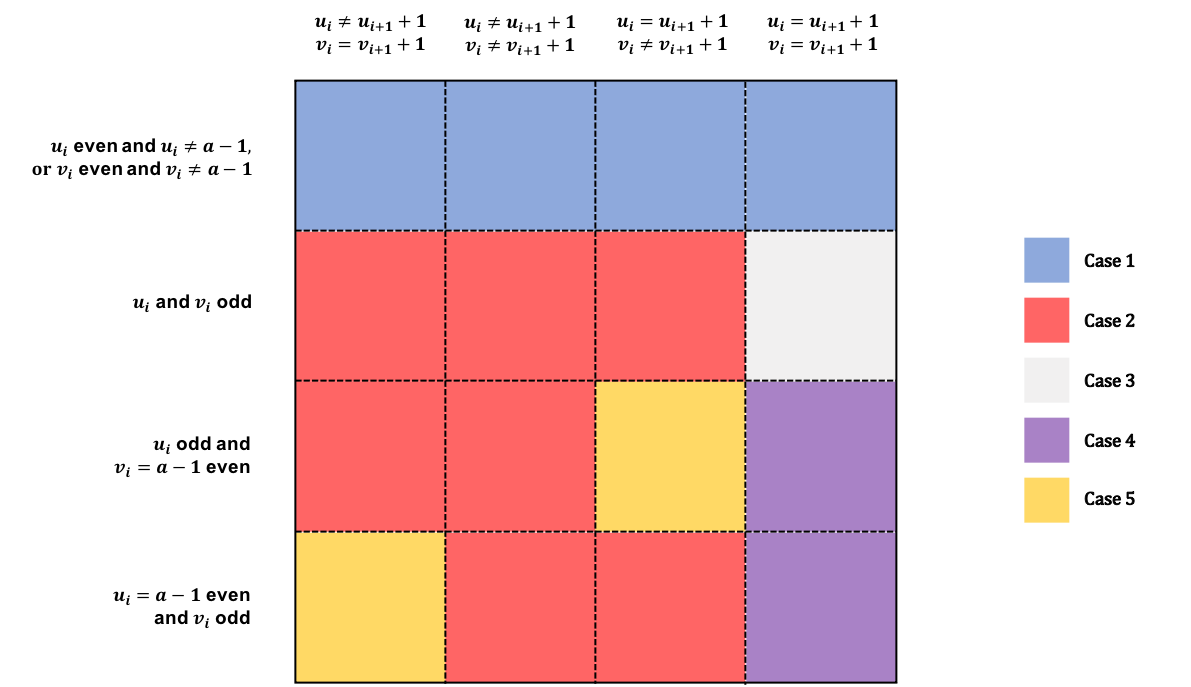}
\caption{Diagram associated with the different cases in the proof of Lemma~\ref{lem:const_red}.}
\label{fig:VennD}
\end{figure}

\textit{Case 1: Without loss of generality assume that $u_i$ even and $u_i\ne(a-1)$.} Define $\alpha:=u_i$; in particular, $\alpha^i(\alpha\!+\!1)^{k+1-i}\in R_{k+1;a}$. We claim that the later string resolves $u$ and $v$. Indeed, we may define
\begin{align*}
\lambda &:= N_\alpha\big(u_{(i-1)}\big) = N_\alpha\big(v_{(i-1)}\big) \\
\gamma &:= N_{\alpha+1}\big(u^{(k-i+1)}\big) = N_{\alpha+1}\big(v^{(k-i+1)}\big).
\end{align*}
Next, using lemmas \ref{lem:two_runs} and~\ref{lem:oneRun} we find that
\[\ell(u,\alpha^i(\alpha\!+\!1)^{k+1-i})\hspace{10cm}\]
\vspace{-24pt}
\begin{align*}
&= \min\{\ell(u_{(i-1)},\alpha^{i}) + \ell(u^{(k-i+1)},(\alpha\!+\!1)^{k+1-i}),\ell(u_{(i)},\alpha^{i}) + \ell(u^{(k-i)},(\alpha\!+\!1)^{k+1-i})\}\\
&\le\ell(u_{(i)},\alpha^{i}) + \ell(u^{(k-i)},(\alpha\!+\!1)^{k+1-i})\\
&= k-\lambda-\gamma,
\end{align*}
where for the second identity we have used that $u_i=\alpha$.  Similarly, using that $v_i\ne\alpha$ we obtain that
\begin{align*}
\ell(v,\alpha^i(\alpha\!+\!1)^{(k+1)-i})&=\min\{k+1-\lambda-\gamma, k+1-\lambda-\gamma+\indicator[v_i=\alpha+1]\}\\
&=k+1-\lambda-\gamma,
\end{align*}
which shows the lemma for the Case 1. 

We emphasize that Case 1 is the only one required for $a=2$. In particular, without any loss of generality we may assume in what remains of this proof that $a\ge3$.

\textit{Case 2: Without loss of generality assume that $u_i\ne u_{i+1}+1$ and that $u_i$ and $v_i$ are odd, or that $u_i$ is odd and $v_i=a-1$ is even.} Define $\alpha:=u_i-1$; in particular, $u_{i+1}\ne\alpha$ and $\alpha^{i+1}(\alpha\!+\!1)^{k-i}\in R_{k+1;a}$. We claim $u$ and $v$ are resolved by the later string. Indeed, preserving the definitions of $\lambda$ and $\gamma$ from Case 1, and using similar arguments to the ones used for that case, we find now that
\[\ell(u,\alpha^{i+1}(\alpha\!+\!1)^{k-i})=\min\{k+2-\lambda-\gamma,k+2-\lambda-\gamma+\indicator[u_{i+1}=\alpha+1]\}=k+2-\lambda-\gamma.\]
On the other hand, note that $v_i\ne\alpha$ otherwise $u_i=a$, which is not possible. Hence, using that $v_i\ne\alpha$ we obtain that
\begin{align*}
\ell(v,\alpha^{i+1}(\alpha\!+\!1)^{k-i}) 
&=\min\{k+1-\lambda-\gamma,\ell(v_{(i+1)},\alpha^i)+\ell(v^{(k-i-1)},\beta^{k-i-1})\}\\
&\le k+1-\lambda-\gamma,
\end{align*}
which shows the lemma for the Case 2.

\textit{Case 3: $u_i$ and $v_i$ odd, $u_i=u_{i+1}+1$, and $v_i=v_{i+1}+1$.} Define $\alpha:=u_i$ and $\beta := \theta(\alpha)$. We claim that $\alpha^i\beta^{k-i-1}\in\theta(R_{k-1;a})$ resolves $u$ and $v$. To show so define 
\begin{align*}
\lambda' &:= N_\alpha(u_{(i-1)}) = N_\alpha(v_{(i-1)})\\
\gamma' &:= N_\beta(u^{(k-i+1)}) = N_\beta(v^{(k-i+1)}).
 \end{align*}
Note that $u_{i+1}\ne\alpha$ and $u_{i+1}\ne\beta$ because $a\ge 3$; in particular, $N_\alpha(u_{(i+1)}) \le i$ and $N_\beta(u^{(k-i)}) \le k-i-1$. As a result, due to lemmas~\ref{lem:two_runs}-\ref{lem:oneRun}, we find that
\[\ell(u,\alpha^i\beta^{k-i-1})\le\ell(u_{(i)},\alpha^i)\!+\!\ell(u^{(k-i)},\beta^{k-i-1})=k-\lambda'-\gamma'-1\]
Likewise:
\[\ell(v,\alpha^i\beta^{k-i-1})\!=\!\min\{\ell(v_{(i)},\alpha^i)\!+\!\ell(v^{(k-i)},\beta^{k-i-1}),\ell(v_{(i+1)},\alpha^i)\!+\!\ell(v^{(k-i-1)},\beta^{k-i-1})\}\!.\]
But note that $N_\alpha(v_{(i+1)})\le i$ because $\alpha$ is odd and $v_{i+1}$ even, and $N_\beta(v^{(k-i)})=N_\beta(u^{(k-i)})\le k-i-1$ because $u^{(k-i+1)}$ and $v^{(k-i+1)}$ are permutations of each other and $u_{i},v_i\ne\beta$. Finally, since $v_i\ne\alpha$ and $v_{i+1}\ne\alpha$, we obtain that
\begin{align*}
\ell(v,\alpha^i\beta^{k-i-1})=\min\{k-\lambda'-\gamma',k-\lambda'-\gamma'+\indicator[v_{i+1} = \beta]\}=k - \lambda'-\gamma,
\end{align*}
which shows the lemma for the Case 3.
        
\textit{Case 4. Without loss of generality assume that $u_i=u_{i+1}+1$ is odd and that $v_i=v_{i+1}+1=a-1$ is even.} In particular, $a$ is odd and $\alpha^{i}\beta^{k-1-i}\in\theta(R_{k-1;a})$ where $\alpha:=u_i$ and $\beta:=\alpha+1$. We claim that $\alpha^{i}\beta^{k-1-i}$ resolves $u$ and $v$. To see this, note that $u_{i+1}\notin\{\alpha,\alpha+1\}$; specifically, $N_\alpha(u_{(i+1)})\le i$ and $N_\beta(u^{(k-i)})\le k-i-1$. So, if  $\lambda'$ and $\gamma'$ are as in Case 3 then Lemma~\ref{lem:two_runs} and Lemma~\ref{lem:oneRun} imply that
\[\ell(u,\alpha^i\beta^{k-i-1})
\le\ell(u_{(i)},\alpha^i)+\ell(u^{(k-i)},\beta^{k-i-1})=k - \lambda'-\gamma'-1.\]
On the other hand, $v_i\ne \alpha$ hence $N_\alpha(v_{(i+1)})\le i$. Additionally, there must be some $v_j=u_{i+1}$ for some $j>i$, so $N_\beta(u^{(k-i)})\le k-i-1$. Thus:
\[\ell(v,\alpha^i\beta^{k-i-1})\hspace{10cm}\]
\vspace{-24pt}
\begin{align*}
&=\min\{\ell(v_{(i)},\alpha^i)+\ell(v^{(k-i)},\beta^{k-i-1}), \ell(v_{(i+1)},\alpha^i)+\ell(v^{(k-i-1)},\beta^{k-i-1})\}\\
&= \min\{k-\lambda'-\gamma'+\indicator[v_i=\beta],k-\lambda'-\gamma'+\indicator[v_i=\beta]-\indicator[v_{i+1}=\alpha]\}\\
&= k - \lambda'-\gamma',
\end{align*}
where for the final identity we have used that $\indicator[v_{i+1}=\alpha]=\indicator[v_{i}=\beta]$. This shows the lemma for the Case 4.

\textit{Case 5. Without loss of generality assume that $u_i=u_{i+1}+1$ is odd and $v_i=a-1\ne v_{i+1}+1$ is even.} In particular, $a$ is odd and $(a-2)^i(a-1)^{k-i-1}\in\theta(R_{k-1;a})$. We claim that $(a-2)^i(a-1)^{k-i-1}$ resolves $u$ and $v$. To show so, define
 \begin{align*}
\lambda'' &:= N_{a-2}(u_{(i-1)}) = N_{a-2}(v_{(i-1)})\\
\gamma'' &:= N_{a-1}(u^{(k-i+1)}) = N_{a-1}(v^{(k-i+1)}).
\end{align*}
Observe that  $0\le u_{i+1}< u_i\le a-2$ so $u_{i+1}\ne a-1$. As a result, due to lemma~\ref{lem:two_runs}-\ref{lem:oneRun}:
\begin{align*}
\ell(u,(a-2)^i(a-1)^{k-i-1})
&\le\ell(u_{(i)},(a-2)^i)+\ell(u^{(k-i)},(a-1)^{k-i-1}\}\\
&=k-\lambda''-\gamma''-\indicator[u_i=a-2]\\
&\le k-\lambda''-\gamma''.
\end{align*}
On the other hand, since $v_i = a-1$, $N_{a-2}(v_{(i+1)}) \le i$. Additionally, $v_{i+1}\ne a-2$. So:
\[\ell(u,(a-2)^i(a-1)^{k-i-1})\hspace{9cm}\]
\vspace{-24pt}
\begin{align*}
&=\min\{\ell(v_{(i)},(a\!-\!2)^i)\!+\!\ell(v^{(k\!-\!i)},(a\!-\!1)^{k\!-\!i\!-\!1}),\ell(v_{(i\!+\!1)},(a\!-\!2)^i)\!+\!\ell(v^{(k\!-\!i\!-\!1)},(a\!-\!1)^{k\!-\!i\!-\!1})\}\\
&=\min\{k-\lambda''-\gamma''+1,k-\lambda''-\gamma''+1+\indicator[v_{i+1}=a-1] \}\\
&=k-\lambda''-\gamma''+1,
\end{align*}
which completes the proof of the lemma.
\end{proof}

\begin{corollary}
$\mathbb{L}_{k_1,k_2,a}$ is resolved by a set of size $O\left(a\big((k_2+1)^2-k_1^2\big)\right)$.
\label{cor:Resolution}
\end{corollary}

\begin{proof}
Let $\theta$ be the character bijection defined in Lemma \ref{lem:const_red}. Consider the sets
\begin{align*}
R_0 &:= \{0^{k_2},\ldots,(a-1)^{k_2}\};\\
R_1 &:= \bigcup_{i=0}^{\lfloor(k_2-k_1)/2\rfloor}\!\!\!\!\!\!\!\!\theta^i(R_{k_2-2i;a})\,\,\cup\,\,
\begin{cases}
\emptyset, & k_2-k_1\ \text{even};\\ 
R_{k_1;a}, & k_2-k_1\ \text{odd}.
\end{cases}
\end{align*}

We claim that $R:=R_0\cup R_1$ resolves $\Lev_{k_1,k_2;a}$. For this, let $u$ and $v$ be different nodes in this Levenshtein graph. We show that $R$ resolves these nodes by considering different cases.

First, suppose that $u$ and $v$ are not permutations of each other; in particular, for some alphabet character $\alpha$, $N_\alpha(u)\ne N_\alpha(v)$. If $|u|=|v|$ then, due to Lemma~\ref{lem:oneRun}, $\ell(u,\alpha^{k_2})=k_2-N_\alpha(u)\ne k_2-N_\alpha(v)=\ell(v,\alpha^{k_2})$ i.e. $u$ and $v$ are resolved. Instead, if $|u|\ne|v|$ and $R_0$ did not resolve them, then 
\begin{align*}
\abs u = \sum_{\alpha=0}^{a-1} N_\alpha(u) = \sum_{\alpha=0}^{a-1} (k_2 - \ell(\alpha^{k_2},u)) = \sum_{\alpha=0}^{a-1} (k_2 - \ell(\alpha^{k_2},v)) = \sum_{\alpha=0}^{a-1} N_\alpha(v) = \abs v,
\end{align*}
which is not possible. Hence $R_0$ resolves all pairs of nodes in $\Lev_{k_1,k_2;a}$ that are not permutations of each other.

Next, suppose that $u\ne v$ are permutations of each other. Let $k:=|u|=|v|$. If $k_2-k$ is even or $k=k_1$ then $\theta^i (R_{k;a})\subset R$ for some integer $0\le i\le\lfloor(k_2-k_1)/2\rfloor$. Further, since $\theta$ is an automorphism, $u_0:=\theta^{-i}(u)$ and $v_0:=\theta^{-i}(v)$ are distinct strings of the same length $k$, and the distances from $u$ and $v$ to the nodes in $\theta^i(R_{k;a})$ is the same as those from $u_0$ and $v_0$ to $R_{k;a}$. But, due to Lemma~\ref{lem:res}, $u_0$ and $v_0$ are resolved by $R_{k;a}$, so $u$ and $v$ are resolved by $\theta^i (R_{k;a})$.

Instead, if $k_2-k$ is odd and $k\ne k_1$ then $\theta^{i+1}(R_{k-1;a})\cup{\theta^{i}}(R_{k+1;a})\subset R$ for some integer $0\le i<\lfloor(k_2-k_1)/2\rfloor$. But $u_0:= \theta^{-i}(u)$ and $v_0:=\theta^{-i}(u)$ are also permutations of each other so, by Lemma~\ref{lem:const_red}, $u_0$ and $v_0$ are resolved by $\theta(R_{k-1;a}) \cup R_{k+1;a}$. Hence, since $\theta$ is an automorphism, $u$ and $v$ are resolved by $\theta^{i+1}(R_{k-1;a})\cup\theta^{i}(R_{k+1;a})$. This shows that $R$ resolves $\Lev_{k_1,k_2;a}$.
    
Finally, observe that 
\begin{align*}
|R_{k,a}|
&=\begin{cases}
1, & \text{if $k=0$};\\
\lfloor \frac{a}{2} \rfloor (k+1), & \text{if $k> 0$}.
\end{cases}
\end{align*}
Therefore
\begin{align*}
|R|
&\le |R_0|+|R_{k_1;a}|+\sum_{i=0}^{\lfloor\frac{k_2-k_1}{2}\rfloor}|\theta^i(R_{k_2-2i;a})|\\
&= a+\left\lfloor\frac{a}{2}\right\rfloor(k_1+1)+\left\lfloor\frac{a}{2}\right\rfloor\sum_{i=0}^{\lfloor\frac{k_2-k_1}{2}\rfloor}(k_2-2i+1)\\
&=O\big(a(k_2+1)(k_2-k_1+1)\big)\\
&=O\left(a\big((k_2+1)^2-k_1^2\big)\right),
\end{align*}
from which the result follows.
\end{proof}

\section{Automorphisms of Levenshtein Graphs}
\label{sec:auto}

In what follows, $\mathbb{A}(G)$ denotes the automorphism group of a graph $G$.

In addition, $\rho$ denotes the string reversal, i.e. if $u=u_1\cdots u_k$ is a string of length $k\ge1$ then $\rho(u):=u_k\cdots u_1$. By definition, $\rho(\varepsilon):=\varepsilon$. On the other hand, given an alphabet bijection $\xi:\{0,\ldots,a-1\}\to\{0,\ldots,a-1\}$, we define $\xi(u):=\xi(u_1)\cdots\xi(u_k)$ and $\xi(\varepsilon):=\varepsilon$. We refer to any such transformation as a character bijection. 

The main result in this section completes the characterization of automorphisms of Levenshtein graphs. The cases not covered by our result have implicitly been addressed in the literature. In fact, $\Lev_{0,1;a}$ is isomorphic to the complete graph $K_{a+1}$, whose automorphism group is the permutation group $S_{a+1}$ (i.e. the set of all permutations of $\{0,\ldots,a\}$). In particular, $|\mathbb{A}(\Lev_{0,1;a})|=(a+1)!$. These Levenshtein graphs are somewhat degenerate in that they are the only Levenshtein graphs where automorphisms do not necessarily preserve string lengths.

On the other hand, $\Lev_{k,k;a}$ is isomorphic to the Hamming graph $\Ham_{k,a}$ (Lemma~\ref{lem:geokka}), whose automorphism group is $(\times_{i=1}^k S_a) \rtimes S_k$~\cite{ChaBer06,TilLla19}. In other words, the automorphisms of $\Lev_{k,k;a}$ are the composition of character permutations with character-wise alphabet bijections. Accordingly, $|\mathbb{A}(\Lev_{k,k;a})|=k!\cdot(a!)^k $. 

The remaining Levenshtein graphs are addressed by our next result.

\begin{theorem}
Let $k_1\ne k_2$ and $k_2\ge2$. In $\Lev_{k_1,k_2;a}$, a node bijection $\sigma$ is an automorphism if and only if $\sigma$ is a character bijection, string reversal, or a composition of both. In particular, $\Lev_{k_1,k_2;a}$ has $a!\cdot2$ automorphisms.
\label{thm:automorphisms}
\end{theorem}

The proof of this theorem is given at the end of this section. It is based on the following five lemmas, and a result from~\cite{Lev01}.

\begin{lemma}
The string reversal and character bijections are automorphisms of $\Lev_{k_1,k_2;a}$.
\label{prop:givenAuto}
\end{lemma}
\begin{proof}
Let $\xi$ be a character bijection. Since $\xi$ and $\rho$ preserve string lengths, $\xi(V_{k_1,k_2;a})\subset V_{k_1,k_2;a}$ and $\rho(V_{k_1,k_2;a})\subset V_{k_1,k_2;a}$. Furthermore, since the character bijection associated with the alphabet bijection $\xi^{-1}$ is an inverse for $\xi$, and $\rho$ is an involution, $\xi$ and $\rho$ are bijections from $V_{k_1,k_2;a}$ onto itself. It is convenient to extend $\xi$ to strings formed from the enlarged alphabet $\{0,\ldots,a-1,\mathtt{-}\}$, defining $\xi(\mathtt{-})=\mathtt{-}$. Likewise, extend $\rho$ to strings that may include indels besides alphabet characters.

Let $u,v\in V_{k_1,k_2;a}$ and $A$ an alignment of length $k\ge1$ between them:
\[A=\begin{array}{ccc}
\alpha_1 & \ldots & \alpha_k\\
\beta_1 & \ldots & \beta_k
\end{array}.\]
Define the following alignment between $\xi(u)$ and $\xi(v)$:
\begin{align*}
\xi(A) &:=\begin{array}{ccc}
\xi(\alpha_1) & \ldots & \xi(\alpha_k)\\
\xi(\beta_1) & \ldots & \xi(\beta_k)
\end{array}.
\end{align*}
Clearly, $\text{score}(\xi(A))=\text{score}(A)$, which implies that $\ell(\xi(u),\xi(v))\le\ell(u,v)$, for all $u,v\in V_{k_1,k_2;a}$ and character bijection $\xi$. In particular, $\ell(\xi^{-1}(\xi(u)),\xi^{-1}(\xi(v)))\le\ell(\xi(u),\xi(v))$, implying that $\ell(u,v)=\ell(\xi(u),\xi(v))$. A similar argument shows that $\ell(u,v)=\ell(\rho(u),\rho(v))$, which completes the proof.
\end{proof}

Next, we discuss the degree of nodes on the infinite graph $\mathbb{L}_a$. Our result can be generalized to arbitrary Levenshtein graphs by restricting the length of the neighbors of a node.

Recall that the number of runs in a node $u$ is denoted $r(u)$.

\begin{lemma}
A node $u$ on $\Lev_a$ has $r(u)$ neighbors of length $\abs u -1$, $\abs u(a-1)$ neighbors of length $\abs u$, and $a+\abs u (a-1)$ neighbors of length $\abs u +1$. In particular, $u$ has degree $a+r(u)+2\abs u (a-1)$.
\label{pro:levDeg}
\end{lemma}
\begin{proof}
Recall that substitutions keep the length of a node, whereas deletions and insertions reduce and increase, respectively, its length by one unit. In particular, $u$ has $|u|(a-1)$ neighbors of length $|u|$, and $r(u)$ neighbors of length $|u|-1$.

Let us now focus on the neighbors of $u$ that can be reached due to a single insertion. An insertion may either keep or increase the number of runs. The former occurs only if a run is enlarged by one character, and there are $r(u)$ ways to do so. The latter occurs only if a run is split by a character into two, or two consecutive runs are separated by a single-character run, which can be done in $(\abs u +1)(a-1) - (r(u)-1)=a+|u|(a-1)-r(u)$ ways. In particular, $r(u) + a+|u|(a-1)-r(u) = a+\abs u (a-1)$ nodes can be reached from $u$ through a single insertion. From this, the proposition follows.
\end{proof}

The number of strings that can be created by a given number of insertions onto a given string, and a bound on the number of strings that can be formed by a given number of deletions from a given string is discussed in~\cite{Lev01}.

\begin{lemma}
If $k_1+1<k_2$ then any automorphism of $\Lev_{k_1,k_2;a}$ preserves the length of strings of length $k_2$.
\label{lem:k2fix}
\end{lemma}
\begin{proof}
Let $\sigma$ be an automorphism of $\Lev_{k_1,k_2;a}$ (recall the implicit assumption that $a\ge2$). We claim that $\sigma(V_{k_2,k_2;a})\subset V_{k_1,k_2-2;a}\cup V_{k_2,k_2;a}$. By contradiction suppose that there is a node $u$ such $|u|=k_2$ and $|\sigma(u)|=k_2-1$. Then, due to Lemma~\ref{pro:levDeg}:
\begin{align*}
\deg(u) &= r(u) + k_2(a-1)\\
\deg(\sigma(u)) &= r(\sigma(u)) + a + 2(k_2-1)(a-1).
\end{align*}
As a result, using that $1\le r(w)\le |w|$ for any non-empty string $w$, we obtain that
\begin{align*}
deg(\sigma(u)) &\ge 1+a+ 2(k_2-1)(a-1)\\
&\ge 1+a+(k_2-1)(a-1)+(k_2-1)\\
&= k_2 + k_2(a-1) + 1\\
&> \deg(u),
\end{align*}
which is not possible because automorphisms preserve node degrees. 

Finally, we show that $\sigma(V_{k_2,k_2;a})=V_{k_2,k_2;a}$. For this note that no vertex in $V_{k_2,k_2-2;a}$ can be a neighbor of a vertex in $V_{k_2,k_2;a}$ because any alignment between a word of length $k_2-2$ and another of length $k_2$ must include at least two indels. On the other hand, since $V_{k_2,k_2;a}$ is the vertex set of $\Ham_{k_2;a}$, which is a connected sub-graph of $\Lev_{k_1,k_2;a}$, $\sigma(V_{k_2,k_2;a})$ is the vertex set of a connected subgraph of $\Lev_{k_1,k_2;a}$. As a result, since $\sigma(V_{k_2,k_2;a})\subset V_{k_1,k_2-2;a}\cup V_{k_2,k_2;a}$, either $\sigma(V_{k_2,k_2;a})\subset V_{k_1,k_2-2;a}$ or $\sigma(V_{k_2,k_2;a})\subset V_{k_1,k_2;a}$. Since the former inclusion is not possible because $|V_{k_1,k_2-2;a}|<|V_{k_1,k_2;a}|$, we must have $\sigma(V_{k_2,k_2;a})\subset V_{k_2,k_2;a}$, which shows the proposition.
\end{proof}

\begin{lemma}
Let $k_1\ne k_2$ and $k_2\ge2$, and define $X:=\{0^{k_2},\ldots,(a-1)^{k_2}\}$. If $\sigma$ is an automorphism of $\Lev_{k_1,k_2;a}$ then $\sigma(X)=X$.
\label{lem:X}
\end{lemma}

\begin{proof}
Let $\sigma$ be an automorphism of $\Lev_{k_1,k_2;a}$. 

We first show that $\sigma(X)\subset V_{k_2,k_2;a}$. Due to Lemma \ref{lem:k2fix}, this is direct when $k_1+1<k_2$. Hence assume that $k_1+1=k_2$; in particular, $V_{k_1,k_2;a}=V_{k_1,k_1;a}\cup V_{k_2,k_2;a}$. Suppose that $\sigma(X)\cap V_{k_1,k_1;a}\ne\emptyset$. Then, there would be $x\in X$ such that $|\sigma(x)|=k_1$. In particular, due to Lemma~\ref{pro:levDeg}, it would follow  that
\begin{align*}
\deg(\sigma(x)) &= a + 2(k_2-1)(a-1)\\
&> a+(k_2-1)(a-1)\\
&= 1+k_2(a-1)\\
&=\deg(x),
\end{align*}
which it is not possible because automorphisms preserve node degrees. As a result, $\sigma(X)\cap V_{k_1,k_1;a}=\emptyset$, i.e. $\sigma(X)\subset V_{k_2,k_2;a}$, which shows the claim.

Finally, since $\sigma(X)\subset V_{k_2,k_2;a}$, for each $x\in X$, Lemma~\ref{pro:levDeg} implies that $\deg(x)=1+k_2(a-1)$ and $\deg(\sigma(x)) = r(\sigma(x)) + k_2(a-1)$. Since $\deg(x)=\deg(\sigma(x))$, we must have $r(\sigma(x))=1$, i.e. $\sigma(x)\in X$, which shows the lemma.
\end{proof}

\begin{lemma}
Let $k_1\ne k_2$ and $k_2\ge2$. If $\sigma$ is an automorphism of $\Lev_{k_1,k_2;a}$ then the following apply.
\begin{enumerate}
\item There is a character bijection $\xi$ such that, for every alphabet character $\alpha$ and string $u\in V_{k_1,k_2;a}$,  $N_{\alpha}(u) = N_{\xi(\alpha)}(\sigma(u))$; in particular, $\sigma(\alpha^k)=\xi(\alpha)^k$ for each alphabet character $\alpha$ and $k_1\le k\le k_2$.
\label{auto_prop_char}
\item For all $u\in V_{k_1,k_2;a}$, $|\sigma(u)| = |u|$.
\label{auto_prop_length}
\item For all $u\in V_{k_1,k_2;a}$ with $|u|=k_2$, $r(\sigma(u)) = r(u)$.
\label{auto_prop_run}
\end{enumerate}
\label{lem:auto_properties}
\end{lemma}

\begin{proof}
Consider an automorphism $\sigma$ of $\Lev_{k_1,k_2;a}$, and let $X$ be as in Lemma \ref{lem:X}. In particular, $\sigma(X)=X$. Since $\sigma$ is bijective, there exists an alphabet bijection $\xi:\{0,\ldots,a-1\}\to\{0,\ldots,a-1\}$ such that $\sigma(x)=\xi(x)^{k_2}$, for each $x\in X$. As before, we denote the automorphism associated with $\xi$ with the same symbol.

Let $\alpha$ be an alphabet character, and $u$ a node in $\Lev_{k_1,k_2;a}$. Since $\alpha^{k_2}\in X$, it follows from Lemma \ref{lem:oneRun} that
\[\ell(\sigma(u),\sigma(\alpha^{k_2})) = \ell(\sigma(u),\xi(\alpha)^{k_2}) = k_2-N_{\xi(\alpha)}(\sigma(u)).\]
Since $\ell(u,\alpha^{k_2}) = k_2 - N_{\alpha}(u)$, and we must have $\ell(u,\alpha^{k_2}) = \ell(\sigma(u),\sigma(\alpha^{k_2}))$, Property \ref{auto_prop_char} follows. From this, Property \ref{auto_prop_length} is immediate because
\[\abs u = \sum_{\alpha=0}^{a-1} N_\alpha(u)=\sum_{\alpha=0}^{a-1} N_{\xi(\alpha)}(\sigma(u))=|\sigma(u)|.\]
Finally, due to Property 2 and Lemma \ref{lem:oneRun}, if $|u|=k_2$ then $\deg(\sigma(u))=r(\sigma(u)) + k_2(a-1)$. Likewise, $\deg(u)=r(u) + k_2(a-1)$. In particular, $r(u)=r(\sigma(u))$ because $\deg(u)=\deg(\sigma(u))$, which shows Property \ref{auto_prop_run}.
\end{proof}

\subsection{Proof of Theorem~\ref{thm:automorphisms}} Let $\sigma$ be an automorphism of $\Lev_{k_1,k_2;a}$, and $\xi$ be the corresponding character bijection described in Lemma \ref{lem:auto_properties}. Observe that $(\xi^{-1}\circ\sigma)$ preserves character counts because, due to property (1) in the lemma, $N_\alpha(u)=N_\alpha((\xi^{-1}\circ\sigma)(u))$ for each character $\alpha$ and $u\in V_{k_1,k_2;a}$. 

Next observe the string $0^{k_2-1}1$. From properties (2) and (3) in Lemma \ref{lem:auto_properties}, we find that $(\xi^{-1}\circ\sigma)(0^{k_2-1}1)$ is a string of length $k_2$ with two runs. In particular, since $(\xi^{-1}\circ\sigma)$ preserves character counts, $(\xi^{-1}\circ\sigma)(0^{k_2-1}1)\in\{0^{k_2-1}1,10^{k_2-1}\}$. If $(\xi^{-1}\circ\sigma)(0^{k_2-1}1) = 10^{k_2-1}$, define $\psi:=\rho$, otherwise define $\psi$ to be the identity. In either case, $\psi$ is its own inverse; in particular, if we define
\[\iota := \psi \circ \xi^{-1} \circ \sigma = \psi^{-1} \circ \xi^{-1} \circ \sigma,\]
then
\begin{equation}
\iota(0^{k_2-1}1)=0^{k_2-1}1.
\label{ide:initcond}
\end{equation}

We aim to show next that $\iota$ is the identity, focusing first on strings of length $k_2$ with two runs. In fact, note that $\iota$ preserves character and run counts because $\psi$ and $(\xi^{-1}\circ\sigma)$ do. Hence, if $\alpha\ne\beta$ are characters and $0<k<k_2$ then 
\begin{equation}
\iota(\alpha^{k_2-k}\beta^k) \in \{\alpha^{k_2-k}\beta^k,\beta^k\alpha^{k_2-k}\}.
\label{ide:1of2}
\end{equation}

First, let $\alpha=0$ and $\beta=1$. Assume that {$\iota(0^{k_2-k}1^k)=1^k0^{k_2-k}$ for some $0<k<k_2$.} Then, using Theorem~\ref{thm:geo_dist}, Corollary~\ref{cor:Lev-Ham}, and Equation~(\ref{ide:initcond}), we find the following distances are
    \begin{align*}
        d(0^{k_2-k}1^k,0^{k_2-1}1) &= h(0^{k_2-k}1^k,0^{k_2-1}1)= k-1;\\
        d(\iota(0^{k_2-k}1^k),\iota(0^{k_2-1}1)) &= h(1^k0^{k_2-k},0^{k_2-1}1)= k+1;
    \end{align*}
which is not possible because automorphisms preserve distances. Thus $\iota(0^{k_2-k}1^k)=0^{k_2-k}1^k$, for all $0<k<k_2$. 

Second, if $\alpha=1$, $\beta=0$, {and $\iota(1^{k_2-k}0^k)=0^k1^{k_2-k}$ for some $0<k<k_2$, then $\iota(1^{k_2-k}0^k) = 0^k1^{k_2-k} = \iota(0^k1^{k_2-k})$,} which is not possible because $\iota$ is one-to-one. {Therefore $\iota(1^{k_2-k}0^k)=1^{k_2-k}0^k$, for all $0<k<k_2$.}

Third, let $\alpha\ne1$ and $\beta = 1$. Assume that $\iota(\alpha^{k_2-k}1^k)\ne\alpha^{k_2-k}1^k$ for some $0<k<k_2$. Then, due to Equation~(\ref{ide:1of2}):
\begin{align*}
d(\alpha^{k_2-k}1^k,0^{k_2-k}1^k) &= h(\alpha^{k_2-k}1^k,0^{k_2-k}1^k)=(k_2-k)\indicator[\alpha\ne0];\\
d(\iota(\alpha^{k_2-k}1^k),\iota(0^{k_2-k}1^k)) &= h(1^k\alpha^{k_2-k},0^{k_2-k}1^k)\\
        &=\begin{cases}
        k_2, & {0<k< k_2/2} \text{ and } \alpha\ne0;\\
        2k, & {0<k< k_2/2} \text{ and } \alpha= 0;\\
        2(k_2-k), & {k_2/2\le k<k_2.}
        \end{cases}
\end{align*}
In particular, $d(\alpha^{k_2-k}1^k,0^{k_2-k}1^k)\ne d(\iota(\alpha^{k_2-k}1^k),\iota(0^{k_2-k}1^k))$, which is a contradiction because $\iota$ must preserve distances. So, $\iota(\alpha^{k_2-k}1^k)=\alpha^{k_2-k}1^k$ {for all $\alpha\ne1$ and $0<k<k_2$.}

Finally, let $\alpha\ne\beta$ be arbitrary characters in the alphabet. {If $\alpha = 1$ let $\gamma=0$, otherwise let $\gamma=1$.} Through our second and third cases we have shown that $\iota(\alpha^{k_2-k}\gamma^k)=\alpha^{k_2-k}\gamma^k$ {for all $0<k<k_2$.} Next, assume that $\iota(\alpha^{k_2-k}\beta^k) \ne \alpha^{k_2-k}\beta^k$ {for some $0<k<k_2$. Then, as we have argued before we find that:}
\begin{align*}
d(\alpha^{k_2-k}\beta^k,\alpha^{k_2-k}\gamma^k) &= h(\alpha^{k_2-k}\beta^k,\alpha^{k_2-k}\gamma^k)=k\indicator[\beta\ne \gamma];\\
d(\iota(\alpha^{k_2-k}\beta^k),\iota(\alpha^{k_2-k}\gamma^k)) 
&= h(\beta^k\alpha^{k_2-k},\alpha^{k_2-k}\gamma^k)\\
&=\begin{cases}
k_2, & k_2/2\le k<k_2\text{ and } \beta \ne \gamma;\\
2(k_2-k), & k_2/2\le k<k_2\text{ and } \beta = \gamma;\\
2k, & 0<k< k_2/2.
\end{cases}
\end{align*}
But then, once again we find that $d(\alpha^{k_2-k}\beta^k,\alpha^{k_2-k}\gamma^k)\ne d(\iota(\alpha^{k_2-k}\beta^k),\iota(\alpha^{k_2-k}\gamma^k))$, which is not possible. Consequently, for all $\alpha\ne\beta$ and $0<k<k_2$, $\iota(\alpha^{k_2-k}\beta^k)=\alpha^{k_2-k}\beta^k$.

Thus far, we have shown that if $u$ is a string where $|u|=k_2$ and $r(u)\le 2$ then $\iota(u)=u$. 

Let $R_{k_2,a}=\{r_1,\ldots,r_n\}$ be as defined by Equation~(\ref{eq:res_Set_Fix_k}). Note, for any $r_i\in R_{k_2,a}$ that $|r_i|=k_2$ and $r(r_i)=k_2$, implying that $\iota(r_i)=r_i$. Further, from Lemma~\ref{lem:res}, the transformation $\Phi(u) := \big(d(u,r_1),\ldots,d(u,r_n)\big)$ {is one-to-one over nodes} of length $k_2$. Consider an arbitrary node $u$ such that $|u| = k_2$. From Theorem \ref{lem:auto_properties}, we know that $|\iota(u)| = k_2$. As a result:
\begin{align*}
\Phi(u) &= \big(d(u,r_1),\ldots,d(u,r_n)\big)\\
&= \big(d(\iota(u),\iota(r_1)),\ldots,d(\iota(u),\iota(r_n))\big)\\
&= \big(d(\iota(u),r_1),\ldots,d(\iota(u),r_n)\big)\\
&= \Phi\big(\iota(u)\big).
\end{align*}
In particular, since $\Phi$ is one-to-one over vectors of length $k_2$, $\iota(u)=u$ for all node $u$ such that $|u|=k_2$.

Finally, we prove by induction $k$, with $k_1\le k\le k_2$, that $\iota(v)=v$ for all $v\in V_{k,k_2;a}$. The base case with $k=k_2$ was just shown above. Next, consider a $k_1\le k<k_2$ and suppose that $\iota(v)=v$, for all $v\in V_{k+1,k_2;a}$. If $k=0$, property \ref{auto_prop_length} of Lemma \ref{lem:auto_properties} implies that $\iota(\epsilon)=\epsilon$; in particular, $\iota(v)=v$ for all $v\in V_{k,k_2;a}$. Instead, if $k>0$, consider a string $u$ of length $k$. From Lemma~\ref{pro:levDeg}, $u$ has $a+|u|(a-1) \ge 3$ neighbors of length $k+1$. Let $v_1,v_2,$ and $v_3$ be different neighbors of $u$ of length $k+1$. By the inductive hypothesis: $\iota(v_i)=v_i$, for $1\le i \le 3$. So, since $\iota$ is an automorphism, $v_1$, $v_2$, and $v_3$ are also neighbors of $\iota(u)$. The end of the proof relies on the following result.

\begin{lemma}
(Adjusted from \cite[Theorem 4]{Lev01}.) 
A node $v$ in $\Lev_a$ is uniquely determined by three of its different neighbors of length $|v|+1$. 
\label{lem:supseq_reconstruction}
\end{lemma}

The lemma implies that $\iota(u)=u$ for all $|u|=k$, i.e. $\iota(v)=v$ for all $v\in V_{k,k_2;a}$.

The above shows that $\iota=\psi^{-1} \circ \xi^{-1} \circ \sigma$ is the identity. In particular, $\sigma=\xi \circ \psi$, where $\xi$ is a character bijection and $\psi$ is either the string reversion or the identity, which completes the proof of Theorem~\ref{thm:automorphisms}.

\section{Determining Number of Levenshtein Graphs}
\label{sec:detSet}

For a graph $G= (V,E)$, a set of nodes $D\subset V$ is called determining when the identity is the only $\sigma\in\mathbb{A}(G)$ such that $\sigma (x) = x$, for all $x\in D$ (this is equivalent to the definition given at the end of the Introduction). The determining number of $G$, denoted $\text{\normalfont Det}(G)$, is the size of its smallest determining set. (A graph with a trivial automorphism group has a determining number of $0$.)

We implicitly encountered determining sets of Levenshtein graphs in the proof of Theorem \ref{thm:automorphisms}, which essentially uses that {$\{0^{k_2},\ldots,(a-1)^{k_2},w\}$, with $w$ any non-palindromic string such that $k_1\le|w|\le k_2$,} is a determining set of $\Lev_{k_1,k_2;a}$ when $k_1\ne k_2$ and $k_2\ge2$.

Since $\Lev_{0,1;a}$ is isomorphic to $K_{a+1}$, it follows from~\cite{Bou06} that $\text{Det}(\Lev_{0,1;a}) = a$. On the other hand, since $\Lev_{k,k;a}$ is isomorphic to $\Ham_{k,a}$, which may be described as the Cartesian product of $k$ copies of $K_a$, tight bounds on $\text{Det}(\Lev_{k,k;a})$ follow from~\cite{Bou09}.

On the other hand, it can be shown by an exhaustive test that if $k_1\ne k_2$ and $(k_2,a)=(2,2)$ then $\text{Det}(\Lev_{k_1,2;2}) = 2>\lfloor a/k_2\rfloor$. In this case, $\{01,00\}$ is one of a few minimal determining sets. Our following result addresses the determining number of the remaining Levenshtein graphs.

\begin{theorem}
If $k_1\ne k_2$, $k_2\ge 2$, and $(k_2,a)\ne (2,2)$ then
\[\text{\normalfont Det}(\Lev_{k_1,k_2;a})=\Big\lceil\frac{a}{k_2} \Big\rceil.\]
\label{thm:detNo}
\end{theorem}

The remainder of this section is devoted to stating and proving two auxiliary results and showing this theorem.

\begin{lemma}
If $k_1\ne k_2$ and $k_2\ge 2$ then at least $(a-1)$ of the $a$ alphabet characters must be represented in a determining set of $\Lev_{k_1,k_2;a}$. 
\label{lem:a-1}
\end{lemma}
\begin{proof}
Let $D=\{d_1,...,d_n\}$, with $n\ge1$, be a determining set, and $S$ the set of alphabet characters that occur at least once in $D$, i.e., $S = \{(d_i)_j : 1\le i \le n, 1\le j \le |d_i|\}$. If $|S|<a-1$ then there would exist at least two distinct alphabet characters $\alpha,\beta\notin S$. Let $\mu$ be the character bijection that swaps $\alpha$ and $\beta$, i.e. $\mu(\alpha)=\beta$ and $\mu(\beta)=\alpha$, but acts as the identity on every other character. Then, $\mu(d)=d$, for all $d\in D$; in particular, since $\mu$ is not the identity, $D$ could not be a determining set. Since this is not possible, $|S|\ge a-1$, which shows the lemma. 
\end{proof}

\begin{lemma}
If $k_1\ne k_2$ and $k_2\ge 2$ then $\text{\normalfont Det}(\Lev_{k_1,k_2;a}) \ge \Big\lceil \frac{a}{k_2} \Big\rceil$.
\label{lem:Detlb}
\end{lemma}

\begin{proof}
Let $D=\{d_1,...,d_n\}$, with $n\ge1$, be a determining set, and $S$ the set of alphabet characters that occur at least once in $D$. Define $\ell_0=0$ and $\ell_i = \sum_{j=1}^i |d_i|$ for $1\le i\le n$.

We claim that $\ell_n\ge a$. By contradiction, assume that $\ell_n < a$. Since $\ell_n\ge |S|$, Lemma~\ref{lem:a-1} implies that $\ell_n=|S|=a-1$. In particular, up to a character bijection, we may assume that $S=\{0,\ldots,a-2\}$, and that $d_i=\ell_{i-1} \ldots (\ell_i - 1)$ for $1\le i\le n$. Consider the character bijection $\mu$ such that $\mu(a-1)=a-1$, and $\mu(j) = \ell_i + \ell_{i-1} - 1 - j$ for $\ell_{i-1}\le j \le \ell_i - 1$ and $1\le i \le n$. In particular, $\mu$ acts as a reversal on each string in $D$. Then $(\mu\circ\rho)(d_i)=d_i$, for all $1\le i\le n$, hence $(\mu\circ\rho)$ must be the identity. However, this is not possible because $(\mu\circ\rho)(0(a-1))=(a-1)(a-2)$. Hence $\ell_n\ge a$, which implies the lemma because $n\cdot k_2 \ge \sum_{i=1}^n|d_i|=\ell_n\ge a$.
\end{proof}

\subsection{Proof of Theorem~\ref{thm:detNo}} Define $n:=\lceil \frac{a}{k_2}\rceil$; in particular, $n\ge1$. Due to Lemma~\ref{lem:Detlb}, it suffices to construct a determining set of size $n$, for which we consider three cases. First, if $k_2\ge a$, define $D:=\{d\}$ where
\begin{equation*}
d :=  \begin{cases}
0^{k_2-1}1, & a=2;\\
0^{k_2-a+2}1\cdots(a-2), & a\ge 3.
\end{cases}
\end{equation*}
Since at least $a-1$ alphabet characters are represented in $d$, the identity is the only character bijection that preserves $d$. On the other hand, if $\sigma=\mu\circ\rho$, where $\mu$ is {any} character bijection then, for $a=2$, $\sigma(d)=\mu(1)\mu(0)^{k_2-1}$ with $k_2-1\ge2$; in particular $\sigma(d)\ne d$. Similarly, if $a\ge3$ then $\sigma(d)=\mu(a-2)\cdots\mu(1)\mu(0)^{k_2-a+2}$ with $k_2-a+2\ge2$, and again $\sigma(d)\ne d$. Therefore, $D$ is a determining set.

Second, if $2<k_2<a$, let $D:=\{d_1,\ldots,d_n\}$ {be of cardinality $n$ such that $d_1 := 0012\ldots(k_2-2)$, $d_1,\ldots,d_n$ are of length $k_2$, and every character in $\{0,\ldots, a-2\}$ is used by at least one node in $D$.} Since $a-1$ alphabet characters are represented in $D$, the identity is the only character bijection that maps each $d_i$ to itself. However, if $\sigma = \mu \circ \rho$, where $\mu$ is {any} character bijection, then $\sigma(d_1) = \mu(k_2-2)\cdots\mu(1)\mu(0)^2\ne d_1$. So, $D$ is a determining set. 

Finally, if $k_2=2$; {in particular, $a\ge3$,} let $D=\{d_1,\ldots,d_n\}$ {be of cardinality $n$ such that $d_1:=01$, $d_2:=12$, $d_1,\ldots,d_n$ are of length $2$, and every character in $\{0,\ldots, a-2\}$ is used by at least one node in $D$.} Once again, since at least $a-1$ alphabet characters are represented in $D$, the identity is the only character bijection that maps each $d_i$ to itself. Next, let $\sigma=\mu\circ \rho$, where $\mu$ is {any} character bijection. If $\sigma(01)=01$ then $\mu(1)=0$. If this is the case then $\sigma(12)=\mu(2)0\ne 12$, i.e. either $\sigma(01)\ne 01$ or $\sigma(12)\ne 12$. Hence $D$ is determining and the theorem follows.


\newpage

\bibliographystyle{siamplain}

\begin{thebibliography}{10}

\bibitem{ArbLanMit02}
{\sc O.~Arbell, G.~M. Landau, and J.~S. Mitchell}, {\em Edit distance of
  run-length encoded strings}, Information Processing Letters, 83 (2002),
  pp.~307 -- 314,
  \url{https://doi.org/https://doi.org/10.1016/S0020-0190(02)00215-6},
  \url{http://www.sciencedirect.com/science/article/pii/S0020019002002156}.

\bibitem{Bou06}
{\sc D.~L. Boutin}, {\em Identifying graph automorphisms using determining
  sets}, The Electronic Journal of Combinatorics,  (2006), pp.~R78--R78.

\bibitem{Bou09}
{\sc D.~L. Boutin}, {\em The determining number of a {C}artesian product},
  Journal of Graph Theory, 61 (2009), pp.~77--87.

\bibitem{ChaBer06}
{\sc F.~A. Chaouche and A.~Berrachedi}, {\em Automorphisms group of generalized
  {H}amming graphs}, Electronic Notes in Discrete Mathematics, 24 (2006), pp.~9
  -- 15, \url{https://doi.org/https://doi.org/10.1016/j.endm.2006.06.003},
  \url{http://www.sciencedirect.com/science/article/pii/S1571065306000047}.
\newblock Fifth Cracow Conference on Graph Theory USTRON '06.

\bibitem{Coo71}
{\sc S.~A. Cook}, {\em The complexity of theorem-proving procedures}, in
  Proceedings of the Third Annual ACM Symposium on Theory of Computing, STOC
  '71, New York, NY, USA, 1971, ACM, pp.~151--158,
  \url{https://doi.org/10.1145/800157.805047},
  \url{http://doi.acm.org/10.1145/800157.805047}.

\bibitem{DurEddKroMit98}
{\sc R.~Durbin, S.~R. Eddy, A.~Krogh, and G.~Mitchison}, {\em Biological
  sequence analysis: probabilistic models of proteins and nucleic acids},
  Cambridge University Press, 1998.

\bibitem{GarJoh79}
{\sc M.~R. Garey and D.~S. Johnson}, {\em Computers and intractability: a guide
  to the theory of NP-completeness}, W. H. Freeman \& Co., New York, NY, USA,
  1979.

\bibitem{GroLes16}
{\sc A.~Grover and J.~Leskovec}, {\em Node2vec: Scalable feature learning for
  networks}, in {Proceedings of the 22nd ACM SIGKDD International Conference on
  Knowledge Discovery and Data Mining}, ACM, 2016, pp.~855--864.

\bibitem{HarMel76}
{\sc F.~Harary and R.~A. Melter}, {\em On the metric dimension of a graph}, Ars
  Combin, 2 (1976), p.~1.

\bibitem{HauSchVie12}
{\sc M.~Hauptmann, R.~Schmied, and C.~Viehmann}, {\em Approximation complexity
  of metric dimension problem}, Journal of Discrete Algorithms, 14 (2012),
  pp.~214 -- 222,
  \url{https://doi.org/https://doi.org/10.1016/j.jda.2011.12.010},
  \url{http://www.sciencedirect.com/science/article/pii/S1570866711001134}.
\newblock Selected papers from the 21st International Workshop on Combinatorial
  Algorithms (IWOCA 2010).

\bibitem{JiaPol19}
{\sc Z.~Jiang and N.~Polyanskii}, {\em On the metric dimension of {C}artesian
  powers of a graph}, Journal of Combinatorial Theory, Series A, 165 (2019),
  pp.~1 -- 14,
  \url{https://doi.org/https://doi.org/10.1016/j.jcta.2019.01.002},
  \url{http://www.sciencedirect.com/science/article/pii/S0097316519300032}.

\bibitem{KhuRagRos96}
{\sc S.~Khuller, B.~Raghavachari, and A.~Rosenfeld}, {\em Landmarks in graphs},
  Discrete Applied Mathematics, 70 (1996), pp.~217 -- 229,
  \url{https://doi.org/https://doi.org/10.1016/0166-218X(95)00106-2},
  \url{http://www.sciencedirect.com/science/article/pii/0166218X95001062}.

\bibitem{Laietal19}
{\sc L.~Laird, R.~C. Tillquist, S.~Becker, and M.~E. Lladser}, {\em
  Resolvability of {H}amming graphs}, arXiv preprint arXiv:1907.05974,  (2019).

\bibitem{Lev66}
{\sc V.~I. Levenshtein}, {\em Binary codes capable of correcting deletions,
  insertions, and reversals}, in Soviet {P}hysics {D}oklady, vol.~10, 1966,
  pp.~707--710.

\bibitem{Lev01}
{\sc V.~I. Levenshtein}, {\em Efficient reconstruction of sequences from their
  subsequences or supersequences}, J. Comb. Theory Ser. A, 93 (2001),
  pp.~310--332, \url{https://doi.org/10.1006/jcta.2000.3081},
  \url{http://dx.doi.org/10.1006/jcta.2000.3081}.

\bibitem{MakUkkNav03}
{\sc V.~M{\"a}kinen, E.~Ukkonen, and G.~Navarro}, {\em Approximate matching of
  run-length compressed strings}, Algorithmica, 35 (2003), pp.~347--369.

\bibitem{NeeWun70}
{\sc S.~B. Needleman and C.~D. Wunsch}, {\em A general method applicable to the
  search for similarities in the amino acid sequence of two proteins}, Journal
  of Molecular Biology, 48 (1970), pp.~443--453.

\bibitem{PisEtDid98}
{\sc N.~Pisanti, E.~Et, and V.~D. Diderot}, {\em Recent duplications in
  genomes: a graph theory approach},  (1998).

\bibitem{Rut21}
{\sc P.~Ruth}, {\em Numerical Encoding of Symbolic Data: Standard, State of the
  Art, and New Techniques}, {Undergraduate Honors Thesis}, University of
  Colorado, March 2021.

\bibitem{Saletal15}
{\sc F.~Sala, R.~Gabrys, C.~Schoeny, and L.~Dolecek}, {\em Three novel
  combinatorial theorems for the insertion/deletion channel}, in 2015 IEEE
  International Symposium on Information Theory (ISIT), IEEE, 2015,
  pp.~2702--2706.

\bibitem{Sla75}
{\sc P.~J. Slater}, {\em Leaves of trees}, Congr. Numer, 14 (1975), p.~37.

\bibitem{Sta11}
{\sc F.~Stahlberg}, {\em Discovering vocabulary of a language through
  cross-lingual alignment}, PhD thesis, Karlsruhe Institute of Technology,
  2011.

\bibitem{TilFroLla19}
{\sc R.~C. Tillquist, R.~M. Frongillo, and M.~E. Lladser}, {\em Metric
  dimension}, Scholarpedia, 14 (2019), p.~53881,
  \url{https://doi.org/10.4249/scholarpedia.53881}.
\newblock revision \#190769.

\bibitem{TilFroLla21}
{\sc R.~C. Tillquist, R.~M. Frongillo, and M.~E. Lladser}, {\em Getting the lay
  of the land in discrete space: A survey of metric dimension and its
  applications}, 2021, \url{https://arxiv.org/abs/2104.07201}.

\bibitem{TilLla19}
{\sc R.~C. Tillquist and M.~E. Lladser}, {\em Low-dimensional representation of
  genomic sequences}, Journal of Mathematical Biology, 79 (2019), pp.~1--29,
  \url{https://doi.org/10.1007/s00285-019-01348-1},
  \url{https://doi.org/10.1007/s00285-019-01348-1}.

\bibitem{Ukk85}
{\sc E.~Ukkonen}, {\em Algorithms for approximate string matching}, Information
  and Control, 64 (1985), pp.~100--118.

\bibitem{VarKusGoy16}
{\sc L.~R. {Varshney}, J.~{Kusuma}, and V.~K. {Goyal}}, {\em On palimpsests in
  neural memory: An information theory viewpoint}, IEEE Transactions on
  Molecular, Biological and Multi-Scale Communications, 2 (2016), pp.~143--153,
  \url{https://doi.org/10.1109/TMBMC.2016.2640320}.

\bibitem{WagFis74}
{\sc R.~A. Wagner and M.~J. Fischer}, {\em The string-to-string correction
  problem}, Journal of the ACM (JACM), 21 (1974), pp.~168--173.

\bibitem{ZhoHeiRay19}
{\sc X.~Zhong, F.~Heinicke, and S.~Rayner}, {\em mi{RB}ase{M}iner, a tool for
  investigating mi{RB}ase content}, RNA biology, 16 (2019), pp.~1534--1546.

\end{thebibliography}

\end{document}